\documentclass[a4paper, reqno]{amsart}

\usepackage{amsmath}
\usepackage{amsthm}
\usepackage{amssymb}
\usepackage{bm}
\usepackage{graphicx}
\usepackage[foot]{amsaddr}

\usepackage{hyperref}

\newcounter{nseccion}[section]

\newtheorem{theorem}{Theorem}[nseccion]
\newtheorem{corollary}[theorem]{Corollary}
\newtheorem{proposition}[theorem]{Proposition}

\newtheorem{definition}[theorem]{Definition}
\newtheorem{lemma}[theorem]{Lemma}
\newtheorem{example}[theorem]{Example}

\makeatletter
\renewcommand\subsubsection{\@secnumfont}{\bfseries}%
\renewcommand\subsubsection{\@startsection{subsubsection}{3}
  \z@{.5\linespacing\@plus.7\linespacing}{-.5em}%
  {\normalfont\bfseries}}
\makeatother

\begin{document}


\title[A conformal boundary for space-times based on light-like geodesics]{A conformal boundary for space-times based on light-like geodesics: the 3-dimensional case} 

\author{A. Bautista$^{1}$}
\email{abautist@math.uc3m.es}
\author{A. Ibort$^{1,2}$}%
\email{albertoi@math.uc3m.es}
\address{\footnotesize $^{1}$Dpto. de Matem\'aticas, Univ. Carlos III de Madrid \protect\\ Avda. de la Universidad 30, 28911 Legan\'es, Madrid, Spain.
}%
\address{\footnotesize $^{2}$ICMAT, Instituto de Ciencias Matem\'{a}ticas (CSIC-UAM-UC3M-UCM)\protect\\ C/ Nicol\'as Cabrera, 13-15, 28049, Madrid, Spain.
}%

\author{J. Lafuente$^{3}$}
\email{lafuente@mat.ucm.es}
\address{\footnotesize $^{3}$Dpto. de Geometr\'{\i}a y Topolog\'{\i}a, Univ. Complutense de
Madrid \protect\\ Avda. Complutense s/n, 28040 Madrid, Spain.
}%

\author{R. Low$^{4}$}
\email{mtx014@coventry.ac.uk}

\address{\footnotesize $^{4}$School of Computing, Electronics and Mathematics, Coventry University \protect\\ Priory Street, Coventry CV1 5FB, UK.}

\begin{abstract}
A new causal boundary, which we will term the $l$-boundary, inspired by the geometry of the space of light rays and invariant by conformal diffeomorphisms for space-times of any dimension $m\geq 3$, proposed by one of the authors (R.J. Low, \textit{The space of null geodesics (and a new causal boundary)}, Lecture Notes in Physics \textbf{692}, Springer, 2006, 35--50) is analyzed in detail for space-times of dimension 3.  
Under some natural assumptions it is shown that the completed space-time becomes a smooth manifold with boundary and its relation with Geroch-Kronheimer-Penrose causal boundary is discussed.   A number of examples illustrating the properties of this new causal boundary as well as a discussion on the obtained results will be provided.
\end{abstract}


\date{January, 2017}

\maketitle 

\section{Introduction}

\setcounter{nseccion}{1}

In order to study a space-time  $M$ in the large, the attachment of a `causal' boundary can be useful. 
There are several boundaries defined in the literature: \emph{Geroch's g-boundary} \cite{Ge68}, \emph{Schmidt's b-boundary} \cite{Sh71}, and the \emph{GKP c-boundary}, called also Geroch-Kronheimer-Penrose's boundary, causal boundary or just $c$-boundary \cite{GKP68}.  Their interest depends on the properties we want to study and their definition being sometimes controversial, though Flores, Herrera and S\'anchez  \cite{FHS11} have provided general arguments that ensure the admissibility of a proposed causal boundary at the three natural levels, i.e., as a point set, as a chronological space and as a topological space with its essential uniqueness stressed.

The development of a topological characterization of causality relations in the space of light rays started by R. Low in \cite{Lo88} (see also \cite{Lo90}, \cite{Lo94}) led the author to a new definition of a causal boundary for a strongly causal space-time  by considering the problem of attaching a future endpoint to a null geodesic $\gamma$ in the space of light rays of the given space-time. The idea behind is to treat all null geodesics which focus at the same point at infinity as the light cone of the (common) future endpoint of these null geodesics \cite{Lo06}.   

The recent contributions in the dual description of causality relations in terms on the geometry and topology of the corresponding spaces of light rays and skies (see for instance \cite{Ch08}, \cite{Ch10}, \cite{Ba14}, \cite{Ba15} and references therein) make  this new notion of causal boundary become more relevant as it can provide, not only an alternative description of the $c$-boundary, but a more suitable way of addressing the overall notion of causal boundary versus the (in general badly behaved) notion of conformal boundary.  Actually the first question raised in  \cite{Lo06}, regarding the proposed new notion of boundary, is if it agrees with GKP $c$-boundary, a question that will be thoroughly addressed below.  We will see that, unfortunately, they are not necessarily  the same in general, but it is easy to find examples in which they are closely related and the set of points where they coincide will be characterized. 

The construction of the new boundary involves determining the limit of the curve of tangent spaces to the skies $S(\gamma(s))$ along the geodesic $\gamma$ in the corresponding Grassmannian manifold (see Sect. \ref{sec:preliminaries} for definitions). Even if such a limit exists because of the compactness of the Grassmannian manifold, it need not be unique, which poses an additional difficulty in the construction of the new boundary.    However in three-dimensional space-times skies are one-dimensional and the corresponding Grassmannian is the projective real line, then the limit exists and is unique which allows an unambiguous  definition of the boundary points.   Thus in this work we will restrict the construction of the new boundary to three-dimensional space-times $M$.   Unexpectedly,  it will be shown that under some natural assumptions the boundary not only carries a natural topology but a smooth structure that makes the extended manifold $\overline{M}$ into a smooth manifold with boundary. As this boundary adds endpoints to the light rays, we will call it the $l$-boundary.

The paper will be organized as follows: in Section~\ref{sec:Low-boundary}, we will accomplish the construction of the  $l$-boundary for $\dim M=3$ and then, in Section \ref{Low-c-boundary} the relation with the causal $c$-boundary will be discussed; it will be checked  that in some simple situations it has good properties. 
We will illustrate the obtained results by collecting some relevant examples in section \ref{sec:examples}. 
Finally, in section \ref{sec:discussion}, the obtained results as well as some open problems will be discussed.


\section{The $l$-boundary for 3--dimensional space-times} \label{sec:Low-boundary}

\setcounter{nseccion}{2}

\subsection{Preliminaries on the spaces of light rays and skies of a space-time}\label{sec:preliminaries}  Let us consider a time-oriented $m$-dimensional conformal Lorentz manifold $(M,\mathcal{C})$ and denote by $\mathcal{N}$ its space of light rays. 
Assuming that $M$ is strongly causal and null pseudo--convex, we ensure that $\mathcal{N}$ is a Hausdorff differentiable manifold \cite[sect. 3]{Lo90b}.

As shown in \cite[Sect. 2.3]{Ba14}, the construction of topological and differentiable structures for the space $\mathcal{N}$  can be achieved by a suitable choice of coordinate charts of subbundles of the tangent bundle $TM$. 
Fixing an auxiliary metric $\mathbf{g}\in \mathcal{C}$, the set $\mathbb{N}^{+}=\{\xi \in TM:\mathbf{g}\left( \xi ,\xi \right) =0,\xi \neq 0,\xi
\,\,\mathrm{future}\}\subset TM$ defines the subbundle of future null
vectors on $M$ and the fibre of $\mathbb{N}^{+}$ at $p\in M$ will be denoted by $\mathbb{N}^{+}_p$. 
Null geodesics defined by two different proportional elements $\xi_1 , \xi_2 \in \mathbb{N}^{+}_p$ have the same image in $M$, and then $\xi_1$ and $ \xi_2$ define the same light ray $\gamma $ in $\mathcal{N}$. 
Since $M$ is assumed to be strongly causal, then for any $p\in M$ there exists a globally hyperbolic, causally convex and convex normal neighbourhood $V\subset M$ with differentiable spacelike Cauchy surface $C$ such that if $\lambda$ is a causal curve passing through $V$, then $\lambda\cap C$ is exactly one point.
Then any light ray $\gamma$ passing through $V$ can be determined by its intersection point with $C$ and a null direction at said point. 
If $\mathbb{N}^{+}\left(C\right)$ is the restriction of the subbundle $\mathbb{N}^{+}$ to the Cauchy surface $C$ then a realization of a coordinate chart at $\gamma\in \mathcal{N}$ can be obtained from a coordinate chart of 
\[
\Omega\left(C\right)=\left\{ v\in \mathbb{N}^{+}\left(C\right):\mathbf{g}\left(v,T\right)=-1 \right\}
\]
where $T\in\mathfrak{X}\left(M\right)$ is a fixed global timelike vector field.

For any point $x\in M$, the set of light rays passing through $x$ is named \emph{the sky of} $x$ and
it will be denoted by $S\left( x\right) $ or $X$, i.e.
\begin{equation}
S\left( x\right) = \{\gamma \in \mathcal{N}: x \in \gamma \subset M\}  = X.
\end{equation}%
Notice that the light rays $\gamma \in S(x)$ are in one-to-one
correspondence with the set of null lines at $T_x M$, hence the sky $S\left( x\right) $ of any point $x\in M$ is diffeomorphic
to the standard sphere $\mathbb{S}^{m-2}$. 
The set of all skies is called the
\emph{space of skies} and defined as
\begin{equation}
\Sigma =\{X\subset \mathcal{N}:X=S\left( x\right) \,\,  \mathrm{ for \, \, some } \,\, x\in M\}
\end{equation}%
and the \emph{sky map} as the application $S:M\rightarrow \Sigma $ that, by \cite[Cor. 17]{Ba15}, is a diffeomorphism when the differentiable structure compatible with the \emph{reconstructive or regular topology} is provided in $\Sigma$ \cite[Def. 1]{Ba14}, \cite[Def. 13]{Ba15}.

An auxiliary metric $\mathbf{g}\in\mathcal{C}$ allows to determine the geodesic parameter for the light ray $\gamma\in \mathcal{N}$ such that $\gamma\left(0\right)\in C$ and $\gamma'\left(0\right)\in \Omega\left(C\right)$. 
So, any curve $\Gamma\subset \mathcal{N}$ corresponds to a null geodesic variation in $M$. Since tangent vectors at $T_{\gamma}\mathcal{N}$ can be defined by tangent vectors $\Gamma ^{\prime }(0)$ of smooth curves $\Gamma:(-\epsilon ,\epsilon )\rightarrow \mathcal{N}$ such that $\Gamma (0)=\gamma $, then the Jacobi field on $\gamma$ of the null geodesic variation defined by $\Gamma$ defines a tangent vector in $T_{\gamma}\mathcal{N}$. 
Since $\Gamma ^{\prime }(0)$ does not depend on the parametrization of the light ray $\gamma$ nor on the auxiliary metric $\mathbf{g}$, then $\eta\in T_{\gamma}\mathcal{N}$ can be identified with an equivalence class of Jacobi fields on $\gamma$ given by 
\[
[J]=J(\mathrm{mod}\gamma ^{\prime })
\]
where $J$ is a Jacobi field along $\gamma$ defined by a null geodesic variation corresponding to a curve $\Gamma:(-\epsilon ,\epsilon )\rightarrow \mathcal{N}$ such that $\Gamma\left(0\right)=\gamma$ and $\Gamma'\left(0\right)=\eta$.
Notice that any Jacobi vector field $J$ defined by a null geodesic variation of $\gamma\in \mathcal{N}$ verifies 
\[
\mathbf{g}\left(J\left(t\right),\gamma'\left(t\right)\right)= \mathrm{constant}
\]
for all $t$ in the domain of $\gamma$.
Abusing the notation, we will also denote simply by $J$ vectors in $T\mathcal{N}$.

A canonical contact structure $\mathcal{H}\subset T\mathcal{N}$ exists in $\mathcal{N}$. 
Although $\mathcal{H}$ can be defined by the canonical 1--form $\theta$ on $T^{*}M$, a description in terms of Jacobi fields can be found at \cite{Lo98}, \cite{Lo06}.
For any $\gamma\in \mathcal{N}$, the hyperplane $\mathcal{%
H}_\gamma \subset T_\gamma\mathcal{N}$ is given by:
\begin{equation}  \label{contact}
\mathcal{H}_{\gamma}=\lbrace J\in T_{\gamma}\mathcal{N}:\mathbf{g}%
\left(J,\gamma ^{\prime }\right)=0 \rbrace \, .
\end{equation}
where $\mathbf{g}\in\mathcal{C}$ is an auxiliary metric defining the parametrization of $\gamma$ such that $\gamma'\left(0\right)\in \Omega\left(C\right)$.

Using the previous description of $T_{\gamma}\mathcal{N}$, if $x\in M$ and $\gamma \in X=S\left(x\right)\in \Sigma$  with $\gamma \left( s_{0}\right) =p$, then
\begin{equation}\label{tangent_sky}
T_{\gamma }X=\{J\in T_{\gamma }\mathcal{N}:J\left( s_{0}\right) =0\left(
\mathrm{{mod}\gamma ^{\prime }}\right) \} \, . 
\end{equation}%
It can be easily seen that if $J\in T_{\gamma }X$, since $\mathbf{g}\left( J,\gamma ^{\prime
}\right) $ is constant and $J\left( s_{0}\right) =0\left( \mathrm{{mod}%
\gamma ^{\prime }}\right) $, then $\mathbf{g}\left( J,\gamma ^{\prime
}\right) =0$ and therefore $T_{\gamma }X\subset \mathcal{H}_{\gamma }$. 
Therefore any $T_{\gamma }X$ is a subspace of $\mathcal{H}_{\gamma }$ and
since $\dim X = m -2$, then $X$ is a Legendrian manifold of the contact structure on $\mathcal{N}$.

The following notation will be used in this paper: if $N$ is a manifold, then its reduced tangent bundle is denoted by $\widehat{T}N$, this is, $\widehat{T}N = \bigcup_{x\in M}\widehat{T}_{x}N$ where $\widehat{T}_x N = T_x N \setminus {0}$.

As indicated in the introduction, in \cite{Lo06} the following new idea for a causal boundary in $M$ is introduced. 
Given a future-directed inextensible null geodesic $\gamma:\left(a,b\right)\rightarrow M$, we can consider the curve $\widetilde{\gamma}:\left(a,b\right)\rightarrow \mathrm{Gr}^{m-2}\left(\mathcal{H}_{\gamma}\right)$ defined by 
$$
\widetilde{\gamma}\left(s\right)= T_{\gamma}S\left(\gamma\left(s\right) \right) \, ,
$$ 
where $S(\gamma(s))$ denotes the sky of the point $\gamma(s)$, that is, the congruence of light rays passing through it.  Notice that the skies $S(p)$ are diffeomorphic to $(m-2)$-dimensional spheres, so $T_{\gamma}S\left(\gamma\left(s\right) \right)$ is contained in the Grassmannian manifold $\mathrm{Gr}^{m-2}\left(\mathcal{H}_{\gamma}\right)$ of $\left(m-2\right)$--dimensional subspaces of $\mathcal{H}_{\gamma}\subset T_{\gamma}\mathcal{N}$.
Defining
\begin{equation}\label{boundary-field}
\begin{tabular}{l}
$\ominus_{\gamma} = \lim_{s\mapsto a^{+}}\widetilde{\gamma}\left(s\right)\in \mathrm{Gr}^{m-2}\left(\mathcal{H}_{\gamma}\right)$ , \vspace{3mm} \\
$\oplus_{\gamma} = \lim_{s\mapsto b^{-}}\widetilde{\gamma}\left(s\right)\in \mathrm{Gr}^{m-2}\left(\mathcal{H}_{\gamma}\right)$ ,
\end{tabular}
\end{equation}
if the previous limits exist, then it is possible to assign endpoints to $\widetilde{\gamma}$. 
The compactness of $\mathrm{Gr}^{m-2}\left(\mathcal{H}_{\gamma}\right)$ assures the existence of accumulation points when $s\mapsto a^{+},b^{-}$. 
If $\ominus_{\gamma}$ and $\oplus_{\gamma}$ exist for any $\gamma\in \mathcal{N}$, they define subsets in $\mathrm{Gr}^{m-2}\left(\mathcal{H}\right)$ but, \emph{a priori}, they do not define a distribution. 
Low defines the points in this new future causal boundary as the classes of equivalence of light rays that can be connected by a curve tangent to some $\oplus_{\gamma}$ at any point \cite{Lo06}.
Analogously, the new past causal boundary is defined by using $\ominus_{\gamma}$.

Now, we will show that, in case of $M$ being $3$--dimensional, this new notion of causal boundary, that will be referred to as the $l$-boundary of $M$ in what follows, have fair topological and differentiable structures.
Observe that in such case $\mathcal{N}$ is also $3$--dimensional since $\dim \mathcal{N} = 2m-3 =3$, and the Grassmannian manifold $\mathrm{Gr}^{m-2}\left(\mathcal{H}\right)$ becomes $\mathrm{Gr}^1\left(\mathcal{H}\right)=\mathbb{P}\left(\mathcal{H}\right)$.


\subsection{Construction of  the $l$-boundary for three-dimensional space-times}\label{sec:construction-Low}
In order to define precisely the $l$-boundary of a space-time, we will construct first a manifold $\widetilde{\mathcal{N}}$ equipped with a regular distribution $\widetilde{\mathcal{D}}$ generated by the tangent spaces of the skies. 
The quotient space $\Sigma^{\sim} = \widetilde{\mathcal{N}} / \widetilde{\mathcal{D}}$ will be shown to be diffeomorphic to $M$. 
Then, assigning endpoints to any $\widetilde{\gamma}\subset \widetilde{\mathcal{N}}$ we will get two distributions $\ominus$ and $\oplus$ in $\mathcal{N}$ whose orbits, under some conditions,  will be identified to points at the boundary of $\widetilde{\mathcal{N}}$. 
Finally, this boundary can be propagated to $M$ via an extension of the diffeomorphism $\Sigma^{\sim} \simeq M$. In this way, the $l$-boundary, as described qualitatively in the last paragraph of the previous section, would be seen now as the orbits of the distributions $\ominus$ and $\oplus$ and it will inherit a differentiable structure.

\subsubsection{Constructing $\widetilde{\mathcal{N}}$}

Let us consider a conformal manifold $\left(M,\mathcal{C}\right)$ where $M$ is $3$--dimensional, strongly causal and null pseudo--convex space-time.   Let us recall that a space-time $M$ is said to be null pseudo-convex \cite{Lo90b} if, given any compact set $K$ in $M$, there is a compact set $K'$ in $M$ such that any null geodesic segment with endpoints in $K$ lies in $K'$ .   Then if follows that $M$ is null pseudo-convex iff $\mathcal{N}$ is Hausdorff (see Prop. 3.2 and ff. in \cite{Lo90b}).  Thus the previous assumption on $M$ being null pseudo-convex is just to ensure that $\mathcal{N}$ is Hausdorff.  Notice that the more conventional assumption of $M$ possessing no naked singularities implies that $\mathcal{N}$ is Hausdorff too, however this condition becomes too strong as it is equivalent to global hyperbolicity, in fact the compactness of the diamonds $J^+(p) \cap J^-(q)$ becomes equivalent to the absence of an inextensible causal curve which lies entirely in the causal future or past of a point\cite{Sa06}.

In this sense it is possible to try to place this property within the causality ladder \cite{Mi08} where it should go immediately below globally hyperbolic spaces.   Examples of strongly causal non null pseudo-convex space-times are provided for instance by Minkowski space-time with a single point removed or Minkowski space-time where a space-like half line has been removed (see Fig. \ref{diapositiva1}).   Notice that the first space is non-causally simple \cite{BE96}, \cite{Mi08}, \cite{Sa06} 
while the second is not only non-causally simple but non-causally continuous too (the illustration displays a non-closed $J^+(p)$)) and it could be conjectured that strongly causal null pseudoconvex space-times are causally simple. 

\begin{figure}[h]
  \centering
    \includegraphics[scale=0.5]{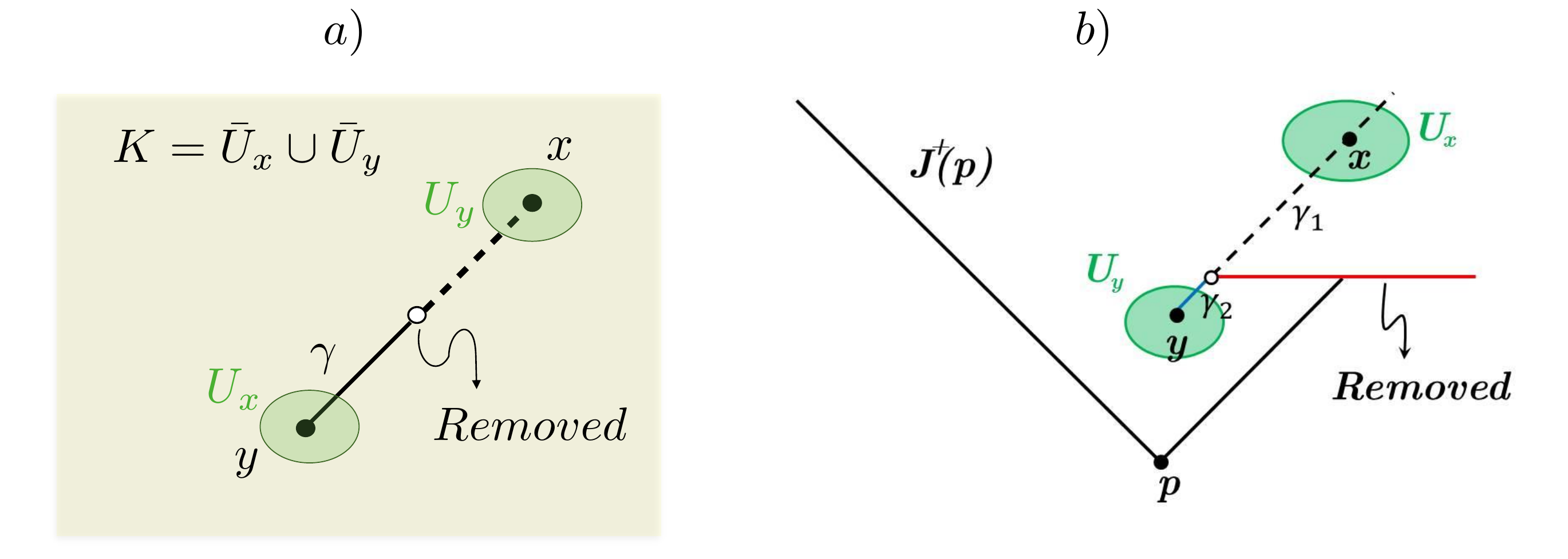}
  \caption{Representation of non null pseudo-convex space-times. a) Minkowski space-time with a single point removed.  There is no compact set containing the compact set $K = \bar{U}_x \cup\bar{U}_y$ and any null  geodesic segment joining pairs of points in $K$. b) Minkowski space-time with a space-like half-line removed.}
  \label{diapositiva1}
\end{figure}

We will restrict in what remains of this section to 3-dimensional space-times, even though many, but not all, arguments and conclusions reached can be extended easily to higher dimensional space-times.  We will use in what follows a particular choice $\mathbf{g}\in\mathcal{C}$ as an auxiliary metric.   Notice that since the projection
$\pi\colon \widehat{T}\mathcal{N} \rightarrow \mathbb{P}\left(T\mathcal{N} \right)$, 
 $J \mapsto \mathrm{span}\left\{J\right\}$, is a submersion, the restriction 
\[
\pi \left. \right|_{\widehat{\mathcal{H}}}:\widehat{\mathcal{H}}\rightarrow \mathbb{P}\left(\mathcal{H}\right) \, ,
\]
 where $\widehat{\mathcal{H}}$ denotes the intersection $\widehat{T}\mathcal{N}\cap \mathcal{H}$, also is so.
Observe that for $X\in \Sigma$ and $J\in T_{\gamma}X$,  we have that $\lambda J\in T_{\gamma}X$ and $\pi\left(\lambda J\right) = \pi\left(J\right)$ for any $\lambda\in \mathbb{R}-\lbrace 0 \rbrace$. 

Let $X\in \Sigma$ be a sky.  Define the map 
\begin{equation}\label{map-tangentsky}
\rho_X \colon X \rightarrow \mathbb{P}\left(\mathcal{H}\right) \, , \qquad \gamma \mapsto T_{\gamma}X \, .
\end{equation}
Let us check that $\rho_X$ is differentiable.   Let $U$ be an open neigborhood of $X$  in the reconstructive topology for $\Sigma$ (see \cite{Ba14}), that is, there is an open set $\mathcal{U} \subset \mathcal{N}$ such that $U = \{ X \in \Sigma \colon X \subset \mathcal{U}\}$.
Restrict the canonical projection $\tau\colon{T\mathcal{N}} \to  \mathcal{N}$ to the regular submanifold $\widehat{T}X\subset \mathcal{H}\left(\mathcal{U}\right)$, where $\mathcal{H}(\mathcal{U})$ denotes the restriction of the bundle $\mathcal{H}$ over $\mathcal{N}$ to the open set $\mathcal{U}$. Consider a differentiable local section $\sigma \colon W\subset X\rightarrow  \widehat{T}X$ of $\left.\tau\right|_{\widehat{T}X}$.
Since any $T_{\gamma}X$ is $1$--dimensional, then $\left.\rho_X\right|_{W} = \left.\pi\right|_{\widehat{T}X} \circ \sigma$ (independently of the section $\sigma$). Then, because $\rho_X|_W$ is the composition of differentiable maps, is differentiable.

Now, we will show that $\rho_X$ is an immersion by proving that it maps regular curves into regular curves.
So, consider any regular curve $\Gamma:I\rightarrow X$. 
The composition of $\Gamma$ with the map in (\ref{map-tangentsky}) gives us the differentiable curve $c=\rho_{X} \circ \Gamma:I\rightarrow \mathbb{P}\left(\mathcal{H}\right)$ defined by $c\left(s\right)=T_{\Gamma\left(s\right)}X$ and since the base curve $\Gamma=\pi \circ c$ is regular then the curve $c$ in the fibre bundle $\mathbb{P}\left(\mathcal{H}\right)$ is also regular.

The image of $\rho_X$ in $\mathbb{P}(\mathcal{H})$ will be denoted as $X^{\sim}= \left\{ T_{\gamma}X:\gamma\in X \right\}$.
\\

Next lemma shows that the union of images $X^{\sim}$ where $X$ lives in any open $U_0\subset \Sigma$ is also open in $\mathbb{P}\left(\mathcal{H}\right)$.

\begin{lemma}\label{open-tilde}
Let $V_0\subset M$ be an open set and $U_0=S\left(V_0\right)\subset \Sigma$. Then $U_0^{\sim}=\bigcup_{X\in U_0}X^{\sim}$ is open in $\mathbb{P}\left(\mathcal{H}\right)$.  
\end{lemma}

\begin{proof}
Given any $P\in U_{0}^{\sim}$ there exist $X\in U_0$ and $\gamma\in X$ such that $P=T_{\gamma}X$. 
Then for this $X\in U_0$, because of  \cite[Thm. 1]{Ba14}, there exists a regular open neighbourhood $U\subset U_0$ of $X$ in $\Sigma$.
This means that the set of vectors $\widehat{U}=\bigcup_{X\in U}\widehat{T}X$ is a regular submanifold in $T\mathcal{U}\subset T\mathcal{N}$ where $\mathcal{U}=\left\{ \gamma\in \mathcal{N}: \gamma \cap S^{-1}\left(U\right)\neq \varnothing \right\}$ (notice that $\gamma \in \mathcal{U}$ if $\gamma$ belongs to some sky $X$ in $U$, but then $X \subset \mathcal{U}$, thus $U$ is the open set corresponding to $\mathcal{U}$ in the reconstructive topology).
Also observe that, since $\mathcal{H}\left(\mathcal{U}\right)= \mathcal{H}\cap T\mathcal{U}$ is a regular submanifold of $T\mathcal{U}$, then $\widehat{U}$ is also a regular submanifold of $\mathcal{H}\left(\mathcal{U}\right)$. 

Because $\dim \widehat{U} = \dim  \mathcal{H}\left(\mathcal{U}\right) = 5$ and $\mathcal{H}\left(\mathcal{U}\right)$ is open in the total space of the bundle $\mathcal{H}$ over $\mathcal{N}$ which has dimension 5 too,  then $\widehat{U}$ is open in $\mathcal{H}\left(\mathcal{U}\right)$ as well as in $\mathcal{H}$.
Since the restriction of the projection $\pi:\mathcal{H}\left(\mathcal{U}\right) \rightarrow \mathbb{P}\left( \mathcal{H}\left(\mathcal{U}\right) \right)$ is a submersion then $\pi\left( \mathcal{H}\left(\mathcal{U}\right) \right)$ is open in $\mathbb{P}\left( \mathcal{H}\left(\mathcal{U}\right) \right)$.
Observe that for $\xi\in T_{\gamma}X$ we have 
\[
\pi\left(\xi\right)=T_{\gamma}X \Longrightarrow \pi\left(\widehat{T}X\right)=X^{\sim} \Longrightarrow \pi\left(\widehat{U}\right)=U^{\sim}
\]
and since $\widehat{U}\subset \mathcal{H}\left(\mathcal{U}\right)$ is open, then $U^{\sim}=\pi\left( \widehat{U} \right)\subset \mathbb{P}\left(\mathcal{H}\left(\mathcal{U}\right)\right)$ is also open, therefore $U^{\sim}$ is open in $\mathbb{P}\left(\mathcal{H}\right)$. 
This shows that $U_0^{\sim}$ is open in $\mathbb{P}\left(\mathcal{H}\right)$.
\end{proof}

The next step is to define the space
\[
\widetilde{\mathcal{N}}=\left\{T_{\gamma}X\in \mathbb{P}\left(\mathcal{H}\right):\gamma\in X\in \Sigma \right\}=\bigcup_{X\in \Sigma}X^{\sim}  \, .
\]
\begin{lemma}\label{Ntilde-open}
$\widetilde{\mathcal{N}}$ is open in $\mathbb{P}\left(\mathcal{H}\right)$.
\end{lemma}

\begin{proof}
If $\left\{ U_{\alpha}\right\}_{\alpha\in \Omega}$ is a open covering of $\Sigma$, then 
\[
\widetilde{\mathcal{N}}=\bigcup_{X\in \Sigma}X^{\sim} = \bigcup_{X\in \bigcup_{\alpha\in \Omega}U_{\alpha} }X^{\sim} = \bigcup_{\alpha\in \Omega}\left(\bigcup_{X\in U_{\alpha}}X^{\sim} \right)
\]
and, by Lemma \ref{open-tilde}, $\widetilde{\mathcal{N}}$ is union of the open sets $U_{\alpha}^{\sim}=\bigcup_{X\in U_{\alpha} }X^{\sim}$, then $\widetilde{\mathcal{N}}$ is open in $\mathbb{P}\left(\mathcal{H}\right)$.
\end{proof}

In order to generalize the present construction to a higher dimensional $M$, it is necessary that $\widetilde{\mathcal{N}}$ be a regular submanifold of $\mathbb{P}\left(\mathcal{H}\right)$. 
This is trivially implied by Lemma \ref{Ntilde-open} in case of a $3$--dimensional $M$ (but not necessarily true in higher dimensions).

\begin{corollary} In a three-dimensional strongly causal and null pseudo-convex conformal space-time,
$\widetilde{\mathcal{N}}$ is a regular submanifold of $\mathbb{P}\left(\mathcal{H}\right)$ that will be called the extended space  of light rays of $M$.
\end{corollary}


\subsubsection{Identifying $M$ inside $\widetilde{\mathcal{N}}$}
We will begin by  expressing the manifold $\widetilde{\mathcal{N}}$ in a different way. 
Let $\gamma:I\rightarrow M$ be an inextensible future-directed parametrized light ray, then we define the curve $\widetilde{\gamma}:I\rightarrow \mathbb{P}\left(\mathcal{H}_{\gamma}\right)$ given by:
\[
\widetilde{\gamma}\left(s\right)=T_{\gamma}S\left(\gamma\left(s\right)\right)\in \mathbb{P}\left(\mathcal{H}_{\gamma}\right) \, ,
\]
and we denote its image by
$\widetilde{\gamma}=\left\{T_{\gamma}S\left(\gamma\left(s\right)\right)\in \mathbb{P}\left(\mathcal{H}_{\gamma}\right):s\in I\right\}$.  
Applying the previous definition of the space $\widetilde{\mathcal{N}}$, it is clear that we can express it in two different ways:
\[
\widetilde{\mathcal{N}}=\bigcup_{X\in \Sigma}X^{\sim} = \bigcup_{\gamma\in \mathcal{N}}\widetilde{\gamma} \, .
\]

It is important to observe that the curve $\widetilde{\gamma}$ is locally injective. 
Indeed, for any $s\in I$ there exists a  globally hyperbolic, causally convex and  normal  convex neighbourhood $V\subset M$ of $\gamma\left(s\right)$. 
This implies that there are no conjugate points in $V$ along $\gamma$, but this also means that for any $t_1,t_2\in I$ such that $\gamma\left(t_i\right) \in V$, $i=1,2$, we have that 
\[
T_{\gamma}S\left(\gamma\left(t_1\right)\right) \cap T_{\gamma}S\left(\gamma\left(t_2\right)\right) = \left\{\mathbf{0}\right\}.
\] 
Therefore it is clear that $T_{\gamma}S\left(\gamma\left(t_1\right)\right)\neq T_{\gamma}S\left(\gamma\left(t_2\right)\right)$.

\begin{definition}
Given a conformal manifold $\left( M,\mathcal{C}\right) $, we will say that 
\begin{enumerate}
\item  $M$ is \emph{null non--conjugate} if for any $x,y\in M$ such that $\gamma\in S\left(x\right) \cap S\left(y\right) \subset \mathcal{N}$ then $T_{\gamma}S\left(x\right) \cap T_{\gamma}S\left(y\right) = \left\{ 0 \right\}$.  
\item $M$ has \emph{tangent skies} if there exist skies $X,Y\in \Sigma$, $X \neq Y$, and $\gamma\in X\cap Y\subset \mathcal{N}$ satisfying $T_{\gamma}X = T_{\gamma}Y$.
\end{enumerate}
\end{definition}

Notice that the notion of null non-conjugate is equivalent to the statement that there are no conjugate points along a null geodesic because if there were a non-zero tangent vector $[J] \in T_{\gamma}S\left(x\right) \cap T_{\gamma}S\left(y\right) $ then, because of (\ref{tangent_sky}), there would be a representative Jacobi field $J$ vanishing at $x$ and $y$ and the points $x$, $y$ would be conjugate.  It is obvious that the null non--conjugate condition automatically implies absence of tangent skies for $M$ of any dimension. 
In the 3--dimensional case, the converse is also true, as it is shown in the following lemma.

\begin{lemma}\label{nullnc-notts}
If $M$ is a $3$--dimensional space-time  without tangent skies then it is also null non--conjugate.  
\end{lemma}

\begin{proof}
Given $X \neq Y\in \Sigma$ with $\gamma\in X \cap Y$ verifying $\widehat{T}_{\gamma}X \cap \widehat{T}_{\gamma}Y \neq \varnothing$, since $\dim T_{\gamma}X = \dim T_{\gamma}Y =1$ then we have $T_{\gamma}X = T_{\gamma}Y$ and therefore $X$ and $Y$ are tangent skies at $M$.
\end{proof}

We have seen that in the 3--dimensional case, $\widetilde{\mathcal{N}}$ is a regular submanifold of $\mathbb{P}\left(\mathcal{H}\right)$. Then if $M$ does not have tangent skies, if $X^\sim \cap Y^\sim \neq \emptyset$ then $T_\gamma X = T_\gamma Y$ for some $\gamma$, then $X = Y$ and $X^\sim = Y^\sim$, hence $\widetilde{\mathcal{N}}$ is foliated by the leaves $X^{\sim}=\left\{ T_{\gamma}X:\gamma\in X \right\}$.   It was proved in \cite{Ba14} that provided that the space-time $M$ is strongly causal and sky-separating (i.e., that the sky map $S$ is injective), there is a basis for the reconstructive topology made of regular open sets, in particular, made of normal open sets where there are no tangent skies (\cite[Defs. 2,3, Thm. 1]{Ba14}).   In \cite{Ba15}, it was also proved that such conditions guarantee that the space of skies with its induced smooth structure is diffeomorphic to $M$, hence we may conclude these remarks by stating that if $M$ is strongly causal and their skies separate points, then the family of regular submanifolds $X^{\sim}$ provide a foliation of $\widetilde{\mathcal{N}}$.
Moreover, since each $X^{\sim}$ is compact, the foliation $\mathcal{D}^{\sim}$ whose leaves are the compact submanifolds $X^\sim$, is regular and the space of leaves:
\[
\Sigma^{\sim} = \widetilde{\mathcal{N}} / \mathcal{D}^{\sim} \, ,
\]
inherits a canonical structure of smooth manifold.

The next proposition gives us the geometric equivalence between $\Sigma^{\sim}$ and its corresponding conformal manifold $M$.  We present it in a general form valid for space-times of dimension higher that 3.

\begin{proposition}\label{conj20}
Let $(M, \mathcal{C})$ be a $m$-dimensional, $m\geq 3$, strongly causal, sky-separating space-time such that the extended space $\widetilde{\mathcal{N}}$ is a regular submanifold of the Grassmannian bundle $\mathrm{Gr}^{m-2}(\mathcal{H})$, then the map $S^{\sim} \colon M \rightarrow  \Sigma^{\sim}$ defined by $S^{\sim}\left(p\right)= S\left(p\right)^{\sim}$ is a diffeomorphism.
\end{proposition}

\begin{proof}
Given a   globally hyperbolic, causally convex and convex normal  open set $V\subset M$, we consider the set of skies $U=S\left(V\right)\subset \Sigma$, the set of vectors $\widehat{U}=\bigcup_{X\in U}\widehat{T}X$ and the set $U^{\sim}=\bigcup_{X\in U}X^{\sim}$. 
By \cite[Thm. 1]{Ba14}, the inclusion $\widehat{U}\hookrightarrow T\mathcal{N}$ is an embedding, and consider the submersion on its range $\pi \colon \mathcal{H}\rightarrow \mathrm{Gr}^{m-2}\left(\mathcal{H}\right)$. 
For $\xi\in T_{\gamma}X\subset \widehat{U}$ then we have that $\pi\left(\xi\right)=T_{\gamma}X$, and then 
\begin{equation}\label{im-tilde}
\pi\left(\widehat{T}X\right)=X^{\sim}
\end{equation}
hence 
\begin{equation}\label{im-tilde-2}
\pi\left(\widehat{U}\right)=U^{\sim}
\end{equation}
So, since $\widehat{U}$ and $U^{\sim}$ are open sets in $\mathcal{H}$ and $\widetilde{\mathcal{N}}$ respectively, it is clear that the restriction $\pi:\widehat{U}\rightarrow U^{\sim}$ is submersion. 
We also know \cite[Thm. 2]{Ba14}, that there exists a regular distribution $\widehat{\mathcal{D}}$ in $\widehat{U}$ whose leaves are $\widehat{T}X=\bigcup_{\gamma\in X}T_{\gamma}X$ with $X\in U$.

Equation (\ref{im-tilde}) implies that there exist a bijection 
\[
\begin{tabular}{rcl}
$\widehat{\pi}: \widehat{U}/ \widehat{\mathcal{D}}$ & $\rightarrow$ & $U^{\sim}/\mathcal{D}^{\sim}$ \\
 $\widehat{T}X$ & $\mapsto$ & $X^{\sim}$
 \end{tabular}
\]
and we obtain the following diagram 
$$
\begin{tabular}{ccc}
$\widehat{U}$ & $\overset{\pi}{\longrightarrow }$ & $U^{\sim}$ \\
$p_1 \downarrow $ &  & $\downarrow p_2$ \\
$\widehat{U}/\widehat{\mathcal{D}}$ & $\underset{\widehat{\pi}}{\rightarrow }$ & $%
U^{\sim}/\mathcal{D}^{\sim}$ \\
\end{tabular}%
$$%
where $p_1$ and $p_2$ are the corresponding quotient maps.
Since $\widehat{\mathcal{D}}$ and $\mathcal{D}^{\sim}$ are regular distributions there exists differentiable structures in $\widehat{U}/ \widehat{\mathcal{D}}$ and $U^{\sim}/\mathcal{D}^{\sim}$ such that $p_1$ and $p_2$ are submersions. 
In this case, $p_2 \circ \pi$ is another submersion, then since both $p_1$ and $p_2 \circ \pi$ are open and continuous, it is clear that the bijection $\widehat{\pi}$ is a homeomorphism.

On the other hand, since $p_1$ is a submersion and $p_2 \circ \pi$ is differentiable, by \cite[Prop. 6.1.2]{BC70}, we have that $\widehat{\pi}$ is differentiable. 
Analogously, since $p_2 \circ \pi$ is a submersion and $p_1$ is differentiable, then $\widehat{\pi}^{-1}$ is differentiable, therefore $\widehat{\pi}$ is a diffeomorphism. 

It is known \cite[Thm. 2]{Ba14} that the quotient $\widehat{U}/ \widehat{\mathcal{D}}$ is diffeomorphic to $V\subset M$ by means of the sky map $S$.
So, we have shown that
$$
\begin{tabular}{rcl}
$S^{\sim} \colon V$ & $\rightarrow$ & $U^{\sim}/\mathcal{D}^{\sim}$ \\
  $p$ & $\mapsto$ & $S^{\sim}\left(p\right)=S\left(p\right)^{\sim}$ 
\end{tabular}
$$
is a diffeomorphism. 

Under the hypothesis of absence of tangent skies, then given $x\neq y\in M$ and $X=S\left(x\right)$, $Y=S\left(y\right)$, we have that $T_{\gamma}X \neq T_{\gamma}Y$, hence $X^{\sim}=S^{\sim}\left(x\right)\neq S^{\sim}\left(y\right)=Y^{\sim}$ implying the injectiveness of the map $S^{\sim} : M \rightarrow  \Sigma^{\sim}$.  
The surjectiveness of  $S^{\sim}$ is obtained by definition, hence it is also a bijection. 
Finally, since $S^{\sim}$ is a bijection and a local difeomorphism at every point, then it is a global diffeomorphism.
\end{proof}


\subsubsection{$\widetilde{\mathcal{N}}$ is a smooth manifold with boundary}
For a parametrized inextensible light ray $\gamma:\left(a,b\right)\rightarrow M$ we define
\begin{equation}\label{Low-field}
\begin{array}{l}
\ominus_{\gamma} = \lim_{s\mapsto a^{+}}\widetilde{\gamma}\left(s\right) \\
\\
\oplus_{\gamma} = \lim_{s\mapsto b^{-}}\widetilde{\gamma}\left(s\right)
\end{array}
\end{equation}
when the limits exist. 

It is clear that if $M$ is 3-dimensional without tangent skies (recall that  in dimension 3 this is equivalent to be non null-conjugate and is automatically satisfied by strongly causal sky separating space-times) then $\tilde{\gamma}$ is injective and its range $\tilde{\gamma}(I) \subset \mathbb{P}\left(\mathcal{H}_{\gamma}\right)\simeq \mathbb{S}^1$, $I= (a,b)$, is an arc-interval in the circle (see Fig. \ref{diapositiva7}), hence there exist the limits in (\ref{Low-field}).   (Notice that in dimension higher than 3, the absence of tangent skies will imply the injectivity of $\tilde{\gamma}$; the compactness of $\mathrm{Gr}^{m-2}(\mathcal{H}_\gamma)$ will guarantee the existence of accumulation points for the set $\tilde{\gamma}(I)$, however this will not suffice to prove the existence of the limits (\ref{Low-field})).
Then under the conditions above it is possible to define the maps
$$
\begin{tabular}{rcl}
$\ominus \colon \mathcal{N}$ & $\rightarrow$ & $\mathbb{P}\left(\mathcal{H}\right)$ \\
 $\gamma$ & $\mapsto$ & $\ominus\left(\gamma\right)=\ominus_{\gamma}$ 
\end{tabular}
\hspace{7mm} \mathrm{and} \hspace{7mm}
\begin{tabular}{rrcl}
$\oplus:$ & $\mathcal{N}$ & $\rightarrow$ & $\mathbb{P}\left(\mathcal{H}\right)$ \\
 & $\gamma$ & $\mapsto$ & $\oplus\left(\gamma\right)=\oplus_{\gamma}$ 
\end{tabular}
$$
and the set 
\[
\overline{\widetilde{\mathcal{N}}}= \bigcup_{\gamma\in \mathcal{N}}\left( \widetilde{\gamma} \cup \left\{ \ominus_{\gamma},\oplus_{\gamma} \right\} \right).
\]

We will analyze now the structure of $\overline{\widetilde{\mathcal{N}}}$ proving that, under natural conditions, it is a smooth manifold with boundary.

First, we will construct local coordinates in $\mathcal{H}$ and $\mathbb{P}\left(\mathcal{H}\right)$ using the ones in $T\mathcal{N}$ defined by the initial values of Jacobi fields at a local Cauchy surface \cite{Ba14}.

Indeed, given a set $V\subset M$ we define $U=S\left(V\right)\subset \Sigma$ and $\mathcal{U}=\bigcup_{X\in U}X\subset \mathcal{N}$.
Let us assume that $V$ is a  globally hyperbolic, causally convex and convex normal  open set in such a way that $\left(V,\varphi =\left( t,x,y \right)\right) $ is a coordinate chart such that the local hypersurface $C\subset V$ defined by $t=0$ is a spacelike (local) Cauchy surface.
Let $\left\{E_{1},E_{2} ,E_{3}\right\} $ be an orthonormal frame in $V$ such that $E_{1} $ is a future oriented timelike vector field in $V$. 
Normalizing the timelike component along $E_1$, writing the tangent vectors of null geodesics at $C$ as $\gamma'\left(0\right)=E_1+u^2 E_2 + u^3 E_3$ and since $\gamma$ is light-like, then $(u^2)^2+(u^3)^2=1$. 
So, we can parametrize all the light rays passing through $\gamma\left(0\right)$ by $u^2=\cos \theta$ and $u^3=\sin \theta$. 
This permits us to define local coordinates in $\mathcal{U}$ by
\[
\psi:\mathcal{U}\rightarrow \mathbb{R}^{3};\hspace{1cm}\psi=\left(x,y,\theta\right)
\]

Moreover, in this case we have that $U\subset \Sigma$ is a regular set in the sense of \cite[Def. 13]{Ba15}, hence $\widehat{U}=\bigcup_{X\in U}\widehat{T}X$ is a regular submanifold of $T\mathcal{U}\subset T\mathcal{N}$ and the inclusion $\widehat{U}\hookrightarrow T\mathcal{N}$ is an embedding.

Consider $\gamma\in \mathcal{U}$ and $J\in T_{\gamma}\mathcal{U}$, since $J$ can be identified with a Jacobi field along the stated parametrization of $\gamma$, we can write $J\left(0\right)=w^1 E_1 + w^2 E_2 +w^3 E_3$ and $J'\left(0\right)=v^1 E_1 + v^2 E_2 +v^3 E_3$. 
Since $\mathbf{g}\left(\gamma', J'\right)=0$ and considering the equivalence $\mathrm{mod}\gamma'$, then denoting $\overline{w}^k=w^k-w^1 u^k$ and $\overline{v}^k=v^k-v^1 u^k$ we have that $\overline{v}^2 u^2 + \overline{v}^3 u^3=0$. 
Supposing without lack of generality that $u^2 \neq 0$ since $\left(u^2, u^3\right)\neq \left(0,0\right)$, we can have $v = \overline{v}^3$, $\overline{w}^2$ and $\overline{w}^3$ as coordinates in $T\mathcal{U}$.
So, we obtain the chart
\[
\overline{\psi}:T\mathcal{U}\rightarrow \mathbb{R}^{6};\hspace{1cm}\overline{\psi}=\left(x,y,\theta,\overline{w}^2,\overline{w}^3,v\right)
\]

Let us define $\mathcal{H}\left(\mathcal{U}\right)=\mathcal{H}\cap T\mathcal{U} = \bigcup_{\gamma\in \mathcal{U}}\mathcal{H}_{\gamma}$.  Now we can construct coordinates in $\mathcal{H}\left(\mathcal{U}\right)\subset T\mathcal{U}$ from $\overline{\psi}$. 
If $J\in \mathcal{H}_{\gamma}$ then $\mathbf{g}\left(\gamma', J\right)=0$ and therefore 
\[
\overline{w}^2 u^2 +  \overline{w}^3 u^3 =0
\]
Again, since $u^2 \neq 0$, we have $\overline{w}^2=-\frac{1}{u^2} \overline{w}^3 u^3$ and we can consider $w=\overline{w}^3$ as a coordinate for $\mathcal{H}\left(\mathcal{U}\right)$, then 
\[
\varphi:\mathcal{H}\left(\mathcal{U}\right)\rightarrow \mathbb{R}^{5} \, ;\qquad \varphi=\left(x,y,\theta,w,v\right)
\]
is a coordinate chart. 

The projection $\pi=\left.\pi^{T\mathcal{N}}_{\mathbb{P}\left(T\mathcal{N} \right)}\right|_{\widehat{\mathcal{H}}}:\widehat{\mathcal{H}}\rightarrow \mathbb{P}\left(\mathcal{H}\right)$ allows us to define coordinates in $\mathbb{P}\left(\mathcal{H}\right)$ as follows. 
From the coordinates $\varphi=\left(x,y,\theta,w,v\right)$, if we consider $J\in \mathcal{H}_{\gamma}$ and $\overline{J}=\lambda J$ for some $\lambda\in \mathbb{R}$, then 
\[
\left\{ 
\begin{array}{l}
\overline{J}\left(0\right) = \lambda J\left(0\right) = \lambda w^1 E_1 + \cdots + \lambda w^m E_m  \\
\overline{J}'\left(0\right) = \lambda J'\left(0\right) = \lambda v^1 E_1 + \cdots + \lambda v^m E_m
\end{array}
\right.
\]
thus the coordinates $w$ and $v$ verify 
\[
\left\{ 
\begin{array}{l}
w\left(\overline{J}\right) = \lambda w\left(J\right)  \\
v\left(\overline{J}\right) = \lambda v\left(J\right) 
\end{array}
\right.
\]
then the homogeneous coordinate $\phi=\left[w:v\right]$ verifies 
\[
\phi\left(\overline{J}\right)=\left[w\left(\overline{J}\right):v\left(\overline{J}\right)\right]=
\left[w\left(J\right):v\left(J\right)\right]=\phi\left(J\right)
\]
and defines the element $\mathrm{span}\left\{J\right\}\in \mathbb{P}\left(\mathcal{H}_{\gamma}\right)$.
Therefore, we obtain that 
\begin{equation}\label{coordinatePH}
\widetilde{\varphi}:\mathbb{P}\left(\mathcal{H}\left(\mathcal{U}\right)\right)\rightarrow \mathbb{R}^{4};\hspace{1cm}\widetilde{\varphi}=\left(x,y,\theta,\phi\right)
\end{equation}
is a coordinate chart in $\mathbb{P}\left(\mathcal{H}\right)$.
Observe that, equivalently, we can also consider $\phi$ as the polar coordinate $\phi=\arctan (w/v)$.

Then we will use local coordinate charts $\left(\mathbb{P}\left(\mathcal{H}\left(\mathcal{U}\right)\right),\widetilde{\varphi}=\left(x,y,\theta,\phi\right)\right)$ as in (\ref{coordinatePH}), where   $\mathcal{U}=\left\{ \gamma \in \mathcal{N}:\gamma\cap V \neq \varnothing \right\}$ is open in $\mathcal{N}$, to describe $\overline{\widetilde{\mathcal{N}}}$ as a manifold with boundary.
In these charts, the coordinate $\phi$ describes the entire $\widetilde{\gamma}$ as well as its limit points. 
Also observe that a light ray $\gamma$ is defined by a fixed $\left(x,y,\theta\right)=\left(x_0,y_0,\theta_0\right)$.

Every fibre $\mathbb{P}\left(\mathcal{H}_{\gamma}\right)$ can be represented by a circumference as shown in Figure \ref{diapositiva7}, where $\widetilde{\gamma}$ is a connected segment of it with endpoints $\ominus_{\gamma}$ and $\oplus_{\gamma}$.

\begin{figure}[h]
  \centering
    \includegraphics[scale=0.25]{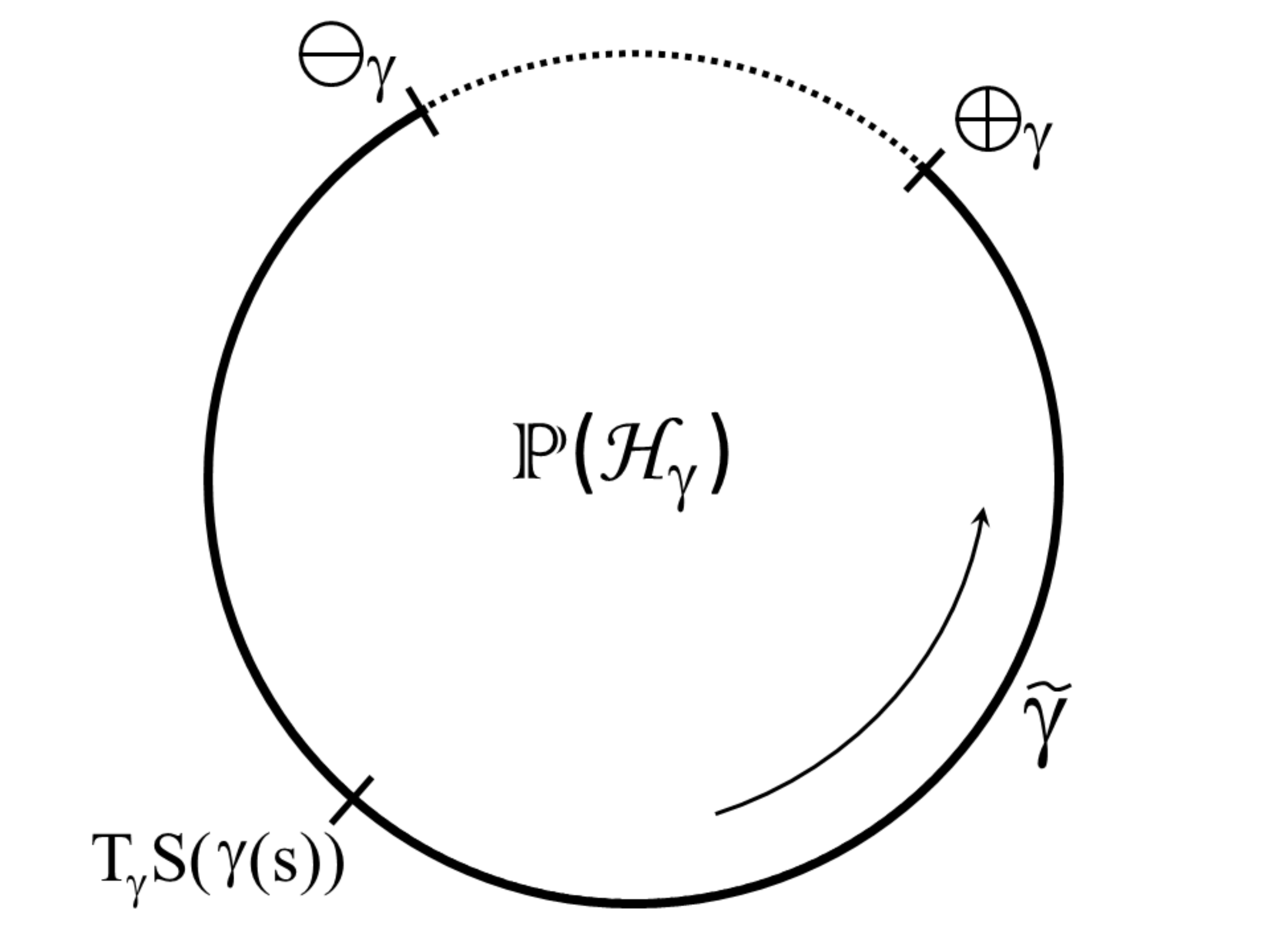}
  \caption{Representation of $\mathbb{P}\left(\mathcal{H}_{\gamma}\right)$.}
  \label{diapositiva7}
\end{figure} 

\begin{proposition}\label{prop-Low-boundary}
Let $M$ be a 3--dimensional null non--conjugate space-time.  Assume that $\ominus$ and $\oplus$ are differentiable distributions.
If $\mathcal{Q}=\left\{ \ominus_{\gamma},\oplus_{\gamma}\in \mathbb{P}\left(\mathcal{H}\right): \ominus_{\gamma}\neq \oplus_{\gamma} \right\}$, then $\overline{\widetilde{\mathcal{N}}}$ is a manifold with boundary the closure $\overline{\mathcal{Q}}$.
\end{proposition}

\begin{proof}
Since $\ominus_{\gamma}$ and $\oplus_{\gamma}$ are defined by the limit of $\widetilde{\gamma}\left(s\right)$ at the endpoints, $\widetilde{\gamma}$ is locally injective and, by Lemma \ref{nullnc-notts}, there are no tangent skies in $M$, then $\widetilde{\gamma}$ must be a connected open set in $\mathbb{P}\left(\mathcal{H}_{\gamma}\right)\simeq \mathbb{S}^1$ with boundary $\left\{\ominus_{\gamma},\oplus_{\gamma}\right\}$.
Now, consider $P\in \mathbb{P}\left(\mathcal{H}\right)$ such that there exist $\gamma\in \mathcal{N}$ verifying $\ominus_{\gamma} = P$ and a coordinate chart $\widetilde{\varphi}=\left(x,y,\theta,\phi\right)$ at $P$ as in (\ref{coordinatePH}). 
Since $\ominus$ is a distribution, for any $\gamma\in \mathcal{N}$ there exists a point $\ominus_{\gamma}\in \mathbb{P}\left(\mathcal{H}_{\gamma}\right)\subset \mathbb{P}\left(\mathcal{H}\right)$ which smoothly depends on the light ray $\gamma$. 
In this case, the coordinates $\left(x,y,\theta\right)$ define the light rays in $\mathcal{N}$, and hence the function $\phi\circ\ominus:\mathcal{N}\rightarrow \left[0,2\pi\right)\simeq \mathbb{S}^1$ depends differentiably on the coordinates $\left(x,y,\theta\right)$.
Analogously, the same rules for $\oplus$.
Let us denote by $\phi_{\ominus}=\phi_{\ominus}\left(x,y,\theta\right)$ and $\phi_{\oplus}=\phi_{\oplus}\left(x,y,\theta\right)$ the coordinate representation of the functions $\phi\circ\ominus$ and $\phi\circ\oplus$ respectively.

Notice that $\partial\overline{\widetilde{\mathcal{N}}}\subset \left\{ \ominus_{\gamma},\oplus_{\gamma}:\gamma\in \mathcal{N}\right\}$. 
 Consider now an open set $\mathcal{U}\subset \mathcal{N}$. 
If $\ominus_{\gamma}\neq\oplus_{\gamma}$ for any $\gamma\in \mathcal{U}$, by locality of $\mathcal{U}$, we can choose, without any lack of generality, a diffeomorphism $\left[0,2\pi\right)\simeq \mathbb{S}^1$ such that 
\[
0<\phi_{\ominus}\left(x,y,\theta\right) < \phi_{\oplus}\left(x,y,\theta\right) < 2\pi
\]
for all $\left(x,y,\theta\right)$ (restricting the domain of $\phi_{\ominus}$ and $\phi_{\oplus}$ if needed). 
Then, for all $\gamma\in \mathcal{U}$, the points in $\overline{\widetilde{\mathcal{U}}}$ can be written as
\[
\overline{\widetilde{\mathcal{U}}}\simeq \left\{ \left(x,y,\theta, \phi\right):\phi_{\ominus}\left(x,y,\theta\right)\leq \phi \leq \phi_{\oplus}\left(x,y,\theta\right) \right\}
\] 
describing a manifold with boundary. 
Then 
\[
\left\{ \ominus_{\gamma},\oplus_{\gamma}: \gamma\in \mathcal{U} \right\} \subset \partial \overline{\widetilde{\mathcal{N}}}
\] 
and, since $\ominus$ and $\oplus$ are regular distributions, the condition $\ominus_{\gamma}\neq \oplus_{\gamma}$ is open in $\mathcal{N}$, therefore we have that 
\[
\mathcal{Q} \subset \partial \overline{\widetilde{\mathcal{N}}}.
\] 

On the other hand, if $\ominus_{\gamma}=\oplus_{\gamma}$ for any $\gamma\in \mathcal{U}$, then we have that $\widetilde{\gamma} \cup \lbrace \ominus_{\gamma} \rbrace = \mathbb{P}\left(\mathcal{H_{\gamma}}\right)$. Again, by the locality of $\mathcal{U}$ then
\[
\mathcal{U}\times\mathbb{S}^1 \simeq \mathbb{P}\left(\mathcal{H}\left(\mathcal{U}\right)\right)=\overline{\widetilde{\mathcal{U}}}
\]
and all the points $\left\{ \ominus_{\gamma}:\gamma\in \mathcal{U}\right\}$ are in the interior of $\overline{\widetilde{\mathcal{U}}}$ and hence, also in the interior of $\overline{\widetilde{\mathcal{N}}}$.

Thus we conclude that $\overline{\mathcal{Q}} \subset \partial \overline{\widetilde{\mathcal{N}}}$.
\end{proof} 

A consequence of the previous proposition is that if $\ominus=\oplus$ then $\overline{\widetilde{\mathcal{N}}}$ is a manifold without boundary.

Notice that the previous result holds if $\ominus$ and $\oplus$ were just continuous distributions. 
In such case, the functions $\phi_{\ominus}$ and $\phi_{\oplus}$ will depend continuously on the coordinates $\left(x,y,\theta\right)$ and the proof would be still valid.


\subsubsection{Constructing the $l$-boundary}
Now, we will see how the $l$-boundary can be assigned to $M$.
Let us now assume for the moment that $\oplus$ and $\ominus$ are regular distributions.
We will split the boundary $\partial \widetilde{\mathcal{N}}$ into the past boundary $\partial^{-} \widetilde{\mathcal{N}}=\left\{ \ominus_{\gamma}:\gamma\in \mathcal{N} \right\}$ and the future boundary $\partial^{+} \widetilde{\mathcal{N}}=\left\{ \oplus_{\gamma}:\gamma\in \mathcal{N} \right\}$.

Let us define the sets of orbits of $\ominus$ and $\oplus$ as
\begin{equation}\label{deltaSigma}
\partial^{-}\Sigma = \mathcal{N}/\ominus \hspace{7mm} \partial^{+}\Sigma = \mathcal{N}/\oplus
\end{equation}
Since $\ominus$ and $\oplus$ are $1$--dimensional distributions, their orbits are $1$--dimensional differentiable submanifolds of $\mathcal{N}$. 
So, for an orbit $X^{+}\in \partial^{+}\Sigma$ and for any $\gamma\in X^{+}$ we have that $T_{\gamma}X^{+}=\oplus_{\gamma}\in \mathbb{P}\left(\mathcal{H}\right)$, and analogously $T_{\gamma}X^{-}=\ominus_{\gamma}\in \mathbb{P}\left(\mathcal{H}\right)$.
This fact implies that the maps
\begin{equation}\label{maps-border-01}
\begin{tabular}{rcl}
 $X^{-}$ & $\rightarrow$ & $\partial^{-}\widetilde{\mathcal{N}}$ \\
 $\gamma$ & $\mapsto$ & $T_{\gamma}X^{-}$ 
\end{tabular}
\hspace{7mm} \mathrm{and} \hspace{7mm}
\begin{tabular}{rcl}
 $X^{+}$ & $\rightarrow$ & $\partial^{+}\widetilde{\mathcal{N}}$ \\
 $\gamma$ & $\mapsto$ & $T_{\gamma}X^{+}$ 
\end{tabular}
\end{equation}
are differentiable because they coincide with the restriction $\left.\ominus\right|_{X^{-}}$ and $\left.\oplus\right|_{X^{+}}$ respectively.

Analogously, we can denote by 
\begin{equation}\label{X+delta}
\left(X^{-}\right)^{\sim}=\left\{ T_{\gamma}X^{-}:\gamma\in X^{-} \right\} \, \quad 
\left(X^{+}\right)^{\sim}=\left\{ T_{\gamma}X^{+}:\gamma\in X^{+} \right\} \, ,
\end{equation}
the corresponding images of the previous maps in (\ref{maps-border-01}).

If $\left(X^{-}\right)^{\sim} \cap \left(Y^{-}\right)^{\sim}\neq \varnothing$ then there exists $\gamma\in X^{-}\cap Y^{-}$ but since  both $X^{-}$ and $Y^{-}$ are orbits of the field of directions $\ominus$ then we have that $X^{-}= Y^{-}$. Analogously for orbits of $\oplus$.
So, we have that the images in $\mathbb{P}\left(\mathcal{H}\right)$ of the orbits of $\ominus$ and $\oplus$ are separate, this means
\[
\left(X^{-}\right)^{\sim} \cap \left(Y^{-}\right)^{\sim}\neq \varnothing \hspace{4mm} \Longrightarrow \hspace{4mm} X^{-}= Y^{-}
\]
\[
\left(X^{+}\right)^{\sim} \cap \left(Y^{+}\right)^{\sim}\neq \varnothing \hspace{4mm} \Longrightarrow \hspace{4mm} X^{+}= Y^{+} \, .
\]
This separation property permits us to define:
\[
\left(\partial^{-}\Sigma\right)^{\sim} =\left\{ \left(X^{-}\right)^{\sim}: X^{-} \in \partial^{-}\Sigma  \right\}
\]
\[
\left(\partial^{+}\Sigma\right)^{\sim} =\left\{ \left(X^{+}\right)^{\sim}: X^{+} \in \partial^{+}\Sigma  \right\} \, ,
\]
and also 
\[
\left(\overline{\Sigma}\right)^{\sim} = \Sigma^{\sim} \cup \left(\partial^{-}\Sigma\right)^{\sim} \cup \left(\partial^{+}\Sigma\right)^{\sim} \, .
\]

Now, observe that the sky map $S^{\sim}:M\rightarrow \Sigma^{\sim}$ in Prop. \ref{conj20}, can be naturally extended to:
\[
\overline{S^{\sim}}:\overline{M} \rightarrow \left(\overline{\Sigma}\right)^{\sim}
\]
by $\overline{S^{\sim}}\left( X^{\pm} \right) = \left( X^{\pm} \right)^{\sim}$,  where $\overline{M} = M \cup \partial^{-}\Sigma \cup \partial^{+}\Sigma$.

\begin{lemma}\label{difeos01}
Under the assumptions stated in this section, the maps: 
$$
\begin{tabular}{rcl}
 $\mathcal{N}$ & $\rightarrow$ & $\partial^{-}\widetilde{\mathcal{N}}$ \\
 $\gamma$ & $\mapsto$ & $\ominus_{\gamma}$ 
\end{tabular}
\hspace{7mm} \mathrm{and} \hspace{7mm}
\begin{tabular}{rcl}
 $\mathcal{N}$ & $\rightarrow$ & $\partial^{+}\widetilde{\mathcal{N}}$ \\
 $\gamma$ & $\mapsto$ & $\oplus_{\gamma}$ 
\end{tabular}
$$
are diffeomorphisms.
\end{lemma}

\begin{proof}
We can see trivially that the map $\mathcal{N}\rightarrow \partial^{-}\widetilde{\mathcal{N}}$ is bijective.
Observe that the image of the map $\ominus:\mathcal{N}\rightarrow \mathbb{P}\left(\mathcal{H}\right)$ is $\partial^{-}\widetilde{\mathcal{N}}$.
Since its expression in coordinates is
\[
\left(x,y,\theta\right)\mapsto \left(x,y,\theta, \phi_{\ominus}\left(x,y,\theta \right)\right)
\]
and $\phi_{\ominus}$ is differentiable, it is clear that $\mathcal{N}$ is locally diffeomorphic to the graph of $\phi_{\ominus}$ and moreover this graph is locally diffeomorphic to the image of $\ominus$, that is $\partial^{-}\widetilde{\mathcal{N}}$.  
So, the map $\mathcal{N}\rightarrow \partial^{-}\widetilde{\mathcal{N}}$ is a bijection and a local diffeomorphism, therefore it is a global diffeomorphism.
The proof for $\mathcal{N}\rightarrow \partial^{+}\widetilde{\mathcal{N}}$ can be done in the same way.
\end{proof}

If $\ominus$ and $\oplus$ define regular distributions in $\mathcal{N}$, we can propagate them to $\partial^{-}\widetilde{\mathcal{N}}$ and $\partial^{+}\widetilde{\mathcal{N}}$ respectively using the difeomorphisms of Lemma \ref{difeos01}.
Then we obtain the regular distributions $\left(\mathcal{D}^{-}\right)^{\sim}$ and $\left(\mathcal{D}^{+}\right)^{\sim}$ on $\partial^{-}\widetilde{\mathcal{N}}$ and $\partial^{+}\widetilde{\mathcal{N}}$
whose leaves are the elements of $\left(\partial^{-}\Sigma\right)^{\sim}$ and $\left(\partial^{+}\Sigma\right)^{\sim}$ respectively.
We will assume in what follows that these distributions, together with the distribution $\mathcal{D}^{\sim}$, give rise to a new distribution $\overline{\mathcal{D}^{\sim}}$  in $\overline{\widetilde{\mathcal{N}}}$.   In other words, it will be assumed that the map assigning to each point $\xi$ in $\overline{\widetilde{\mathcal{N}}}$ the corresponding subspace $\mathcal{D}^{\sim}_\xi$ if $\xi \in \widetilde{\mathcal{N}}$, or $\left(\mathcal{D}^{\pm}\right)^{\sim}_\xi$ if $\xi \in \partial^{\pm}\widetilde{\mathcal{N}}$, is smooth.     

The leaves of $\overline{\mathcal{D}^{\sim}}$ are disjoint in $\overline{\widetilde{\mathcal{N}}}$ and they can be seen as elements of $\left(\overline{\Sigma}\right)^{\sim}$.
Since all the distributions $\mathcal{D}^{\sim}$, $\left(\mathcal{D}^{-}\right)^{\sim}$ and $\left(\mathcal{D}^{+}\right)^{\sim}$ are regular, then $\overline{\mathcal{D}^{\sim}}$ is also a regular distribution.
Therefore we can consider the quotient 
\begin{equation}\label{boundary-chain}
\overline{\widetilde{\mathcal{N}}} / \overline{\mathcal{D}^{\sim}} = 
\widetilde{\mathcal{N}} / \mathcal{D}^{\sim} \cup \partial^{-} \widetilde{\mathcal{N}} / \left(\mathcal{D}^{-}\right)^{\sim} \cup \partial^{+} \widetilde{\mathcal{N}} / \left(\mathcal{D}^{+}\right)^{\sim}
\end{equation}
as a differentiable manifold that, in virtue of Lemma \ref{difeos01}, [\ref{deltaSigma}] and [\ref{X+delta}],  can be identified with: 
\[
\left(\overline{\Sigma}\right)^{\sim} = \Sigma^{\sim} \cup \left(\partial^{-}\Sigma\right)^{\sim} \cup  \left(\partial^{+}\Sigma\right)^{\sim} \simeq \overline{\widetilde{\mathcal{N}}} / \overline{\mathcal{D}^{\sim}}
\]
whose boundary is: $\partial\left(\overline{\Sigma}\right)^{\sim} = \left(\partial^{-}\Sigma\right)^{\sim} \cup \left(\partial^{+}\Sigma\right)^{\sim}$.

Then we can identify $\left(\overline{\Sigma}\right)^{\sim}$ with $\overline{M}$ via the map $\overline{S^{\sim}}:\overline{M} \rightarrow \left(\overline{\Sigma}\right)^{\sim}$, obtaining that $\overline{M}$ is the causal completion we were looking for. 
We state that the $l$-boundary of $M$ is 
\[
\partial_l M = \overline{M}-M = \partial^{-}\Sigma \cup \partial^{+}\Sigma
\] 

In case of $\ominus = \oplus$ then $\partial^{+}\widetilde{\mathcal{N}}=\partial^{-}\widetilde{\mathcal{N}}$
and $\left(\partial^{+}\Sigma \right)^{\sim}=\left(\partial^{-}\Sigma \right)^{\sim}$. 
Hence $\left(\mathcal{D}^{+}\right)^{\sim}=\left(\mathcal{D}^{-}\right)^{\sim}$ and $\partial^{-}\Sigma = \partial^{+}\Sigma$ and therefore, the $l$-boundary of $M$ is 
\[
\partial_l M = \overline{M}-M = \partial\Sigma 
\]
where $\partial\Sigma =\partial^{-}\Sigma = \partial^{+}\Sigma$.  Notice that in such situation $\overline{M}$ is a manifold without boundary.

Collecting the results described in the previous sections we may state the following proposition:

\begin{proposition}  Let $M$ be a strongly causal, sky-separating, 3-dimensional space-time and $\widetilde{\mathcal{N}}$
its extended space of light rays.  Assuming that the limiting distributions $\oplus$, $\ominus$ are regular and extend smoothly the canonical distribution $\mathcal{D}^\sim$ to the boundary of the manifold $\widetilde{\mathcal{N}}$, defining in this way a regular distribution $\overline{\mathcal{D}^\sim}$ of $\overline{\widetilde{\mathcal{N}}}$, then the $l$-boundary $\partial_lM$ of $M$ is well defined, and $\overline{M} = M \cup \partial_lM$ is a smooth manifold with boundary that can be identified naturally with the leaves of the distribution $\overline{\mathcal{D}^\sim}$.
\end{proposition}

Notice that the strong causality and sky-separating conditions stated in the proposition imply that the space $M$ has no tangent skies, hence there are no null-conjugate points, then the boundary of the extended space of light rays is well defined and is smooth.  Moreover if $M$ is null pseudo-convex then the space of light rays is Hausdorff as well as its closure and, because of the assumption on the regularity of the distributions, the quotient will be Hausdorff too.



\section{Comparison with the causal $c$-boundary}\label{Low-c-boundary}

\setcounter{nseccion}{3}

The classical definition of $c$-boundary has been redefined along the years to avoid the problems arising in the study of its topology.  For our purposes, we will recall and deal with its classical definition, but the reader may consult \cite{FHS11}, \cite{Sa09} and references therein,  to get a wider understanding on the subject.

\begin{definition}
A set $W\subset M$ is said to be an \emph{indecomposable past set}, or an \emph{IP}, if it verifies the following conditions:
\begin{enumerate}
\item \label{IP-1} $W$ is open and non--empty.
\item \label{IP-2} $W$ is a \emph{past set}, that is $I^{-}\left( W\right) = W$.
\item $W$ cannot be expressed as the union of two proper subsets satisfying conditions \ref{IP-1} and \ref{IP-2}.
\end{enumerate}
We will say that an IP $W$ is a \emph{proper IP}, or \emph{PIP}, if there is $p\in M$ such that $W= I^{-}\left(p\right)$. In other case, $W$ will be called a \emph{terminal IP} or \emph{TIP}. 
In an analogous manner, considering the chronological future, we can define \emph{indecomposable future sets} or \emph{IF}, then we obtain \emph{proper IFs} and \emph{terminal IFs}, that is, \emph{PIFs} and \emph{TIFs}. 
\end{definition}

In Figure \ref{diapositiva4}, as shown in \cite[Fig. 6.4]{BE96}, a trivial example of the identification of IPs and IFs with boundary points of $M$ is offered. 
We consider $M$ a cropped rectangle of the $2$--dimensional Minkowski space-time  equipped with the metric $\mathbf{g}%
=-dy\otimes dy+dx\otimes dx$. Points at the boundary of $M$ such as $p$ are related to TIPs like $A$, those such as $q$ corresponds to TIFs like $B$ and those such as $r$ can be related to TIPs like $C$ as well as TIFs like $D$.

\begin{figure}[h]
  \centering
    \includegraphics[scale=0.25]{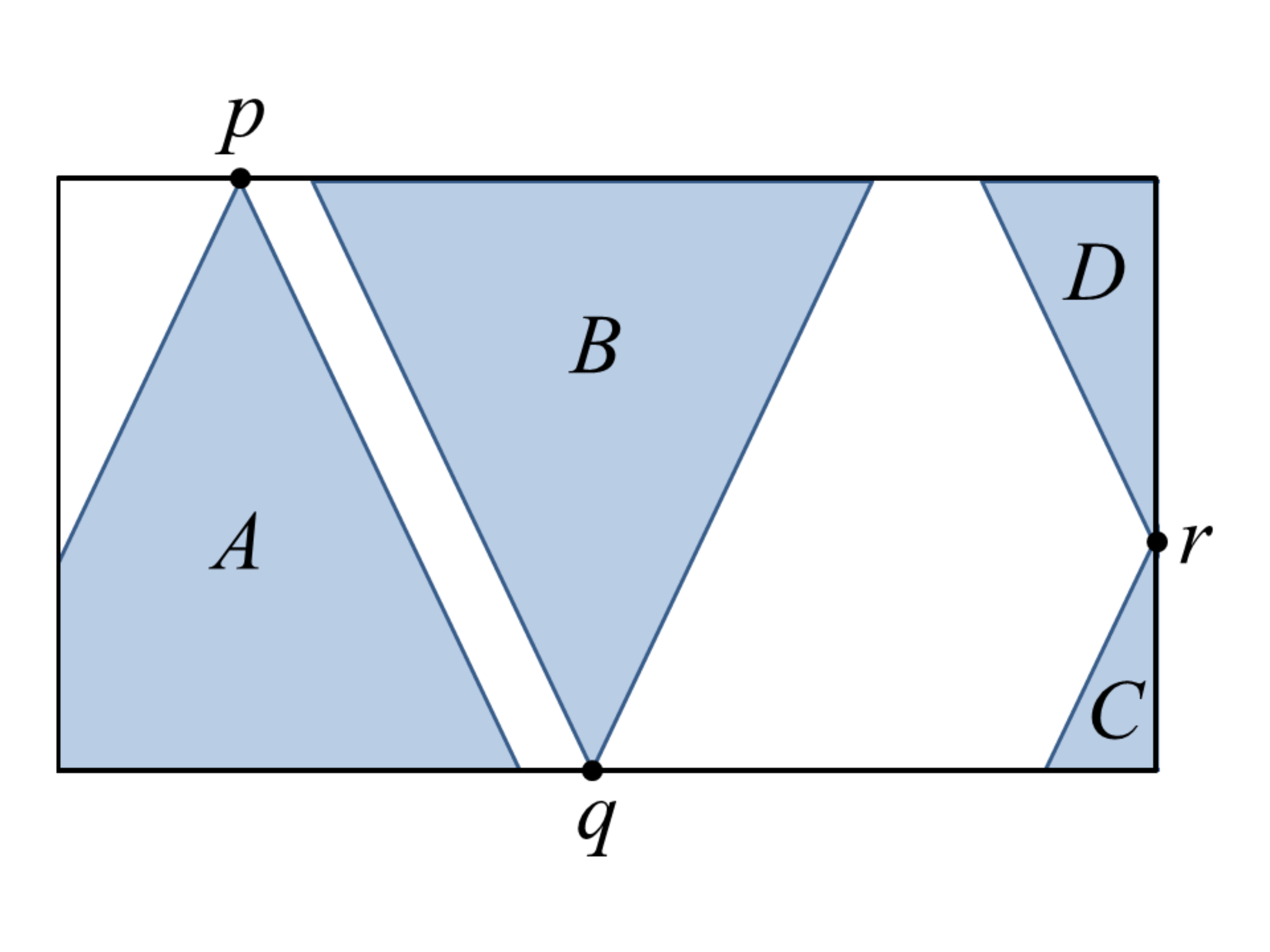}
  \caption{TIPs and TIFs.}
  \label{diapositiva4}
\end{figure}

The following proposition provide us a characterization of all TIPs in a strongly causal space-time.

\begin{proposition}
For any strongly causal space-time  $M$, $A\subset M$ is a TIP if and only if there exists an inextensible to the future timelike curve $\mu$ such that $A=I^{-}\left(\mu\right)$.
\end{proposition}

\begin{proof}
See \cite[Prop. 6.8.1]{HE73}.
\end{proof}

Light rays also define terminal ideal points as next proposition shows.

\begin{proposition}
Let $\gamma$ be a future--directed inextensible causal curve in a strongly causal space-time  $M$, then $I^{-}\left(\gamma\right)$ is a TIP.
\end{proposition}

\begin{proof}
See \cite[Prop. 3.32]{FHS11}.
\end{proof}

Now, we are ready for the classical definition of GKP $c$-boundary.

\begin{definition}
We define the \emph{future (past) causal boundary}, or \emph{future (past) $c$-boundary} of $M$, as the set of all TIPs (TIFs).
\end{definition}

Observe that any point $p\in M$ can be identified with the PIP $I^{-}\left(p\right)$ as well as the PIF $I^{+}\left(p\right)$, moreover it is possible that there exist a TIP and TIF identified with the same point at the boundary (as TIP $C$ and TIF $D$ in Figure \ref{diapositiva4}).
Then, in order to define the causal completion of $M$, a suitable identification between sets of IPs and IFs is needed. 
This is beyond the scope of this work, but \cite{FHS11} and its references can be consulted for further information.

The question arising now is if all TIPs in the future $c$-boundary can be defined by the chronological past of a light ray. 
Unfortunately, this is not always true because there may be TIPs that can only be defined by time-like curves as the following example shows and which implies that the $c$-boundary and $l$-boundary are different in general.
We will denote by $I^{\pm}\left(\cdot , V\right)$ the chronological relations $I^{\pm}\left(\cdot\right)$ restricted to $V$. 
It is clear that $I^{\pm}\left(\cdot , V\right) \subset I^{\pm}\left(\cdot\right) \cap V$, but  equality does not always hold. 

\begin{example}\label{example-Low-not-GKP}  A simple example comparing the $c$-boundary and the $l$-boundary.

Let $\mathbb{M}^{3}$ be the $3$--dimensional Minkowski space-time  and $\mathcal{N}$ its space of light rays. 
Let us choose any point $\omega\in \mathbb{M}^{3}$ and consider the space-time  $M$ as the restriction of $\mathbb{M}^{3}$ to any open half $K\subset \mathbb{M}^{3}$ of a solid cone with vertex in $\omega$ such that $K\subset I^{-}\left(\omega\right)$, as figure \ref{diapositiva5} shows. 
Notice that $M=I^{-}\left(\omega\right)$ can also be considered.
Observe that there exists a light ray $\gamma$ arriving at points like $p^{*}$, so a point $X^{+}_{\gamma}\in \partial^{+}\Sigma_M$ can be defined by $\gamma$, and notice that $p^{*}$ can be identified with the TIP $I^{-}\left(\gamma,M \right)$.
But also observe that the point $\omega$ is not accessible by any light ray in $M=K$ so there is no point in the future  $l$-boundary corresponding to the TIP $M=I^{-}\left(\mu, M\right)$ defined by the future--inextensible timelike curve $\mu$ ending at $\omega$ shown in the picture.

\begin{figure}[h]
  \centering
    \includegraphics[scale=0.25]{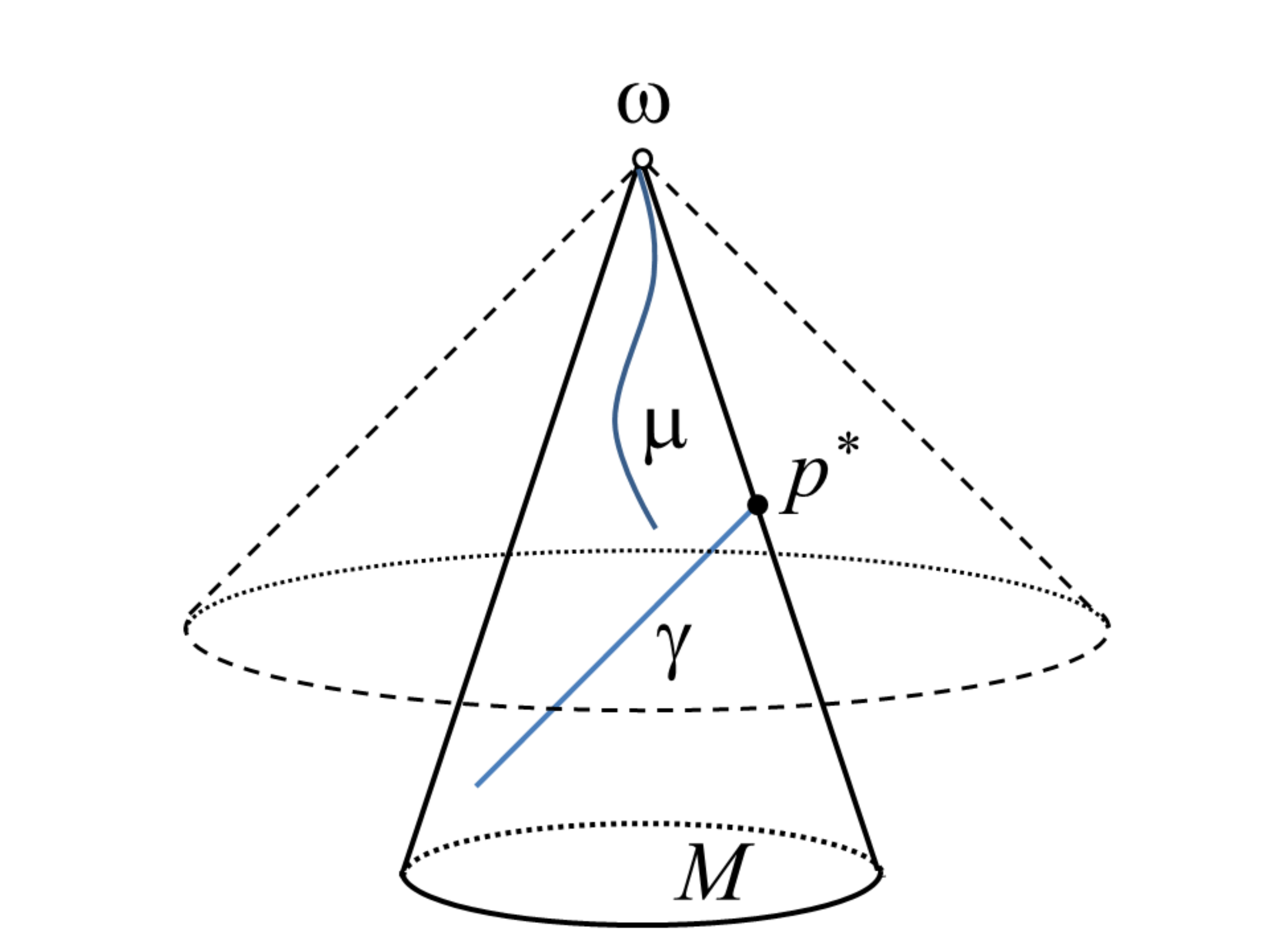}
  \caption{The $l$-boundary is not GKP.}
  \label{diapositiva5}
\end{figure}

\end{example}

However in spite of the previous example, we can see that the $l$-boundary is closely related to the GKP $c$-boundary when we include some topological constraints to the space-time.   The considerations to follow apply in any dimension provided that the limiting distributions $\oplus$, $\ominus$ exist (similarly as was remarked previously in Sect. \ref{sec:construction-Low} in various occasions) and unless stated explicitly we will not be restricted to the 3-dimensional setting.

As a first step, it is possible to study the $l$-boundary corresponding to the restriction of a space-time  $M$ to a suitable open set $V\subset M$. 
The aim of it is to know how to identify $\partial \Sigma$ under na\"{i}ve conditions. 
The study of the future $l$-boundary $\partial^{+} \Sigma$ is enough for this purpose because the past one is analogous.  

Consider $V\subset M$ a relatively compact,  globally hyperbolic, causally convex and convex normal  open set and $\mathcal{U}=\left\{\gamma\in \mathcal{N}:\gamma\cap V\neq \varnothing\right\}$. 
We denote by $\oplus^{V}$ the field of limiting subspaces tangent to the skies of points in a future-directed light ray when they tend to the future boundary of $V$ that, as indicated before, will be assumed to exist  (later on we will discuss a situation where the existence of the limit will be guaranteed). 
So, given $\gamma\in \mathcal{U}\subset \mathcal{N}$ we can give a  future--directed parameterization of  the segment of $\gamma$ in $V$ by $\gamma:\left(a,b\right)\rightarrow V$. Then: 
\[
\oplus^{V}_{\gamma}=\oplus^{V}\left(\gamma\right)=\lim_{s\mapsto b^{-}}T_{\gamma}S\left(\gamma\left(s\right)\right)
\]

Observe that a curve $c:I\rightarrow \mathcal{U}$ is the integral curve of $\oplus^{V}$ passing through $\gamma$ at $\tau=0$ if 
\[
\left\{
\begin{array}{l}
 c'\left(\tau\right)\in \oplus^{V}\left(c\left(\tau\right)\right) \\
c\left(0\right)=\gamma 
\end{array}
\right.
\] 

Now, consider $x\in \partial V\subset M$ such that $\lim_{s\mapsto b^{-}}\gamma\left(s\right) = x$ and let $\Gamma:I\rightarrow X\cap\mathcal{U}$ be a curve travelling along the light rays of the sky $X=S\left(x\right)$ in $\mathcal{U}$ such that $\Gamma\left(\tau\right)=\gamma_{\tau}$ with $\gamma_0=\gamma$ and $\gamma_{\tau} \cap \overline{V}$ has a future endpoint at $x$ for all $\tau\in I$.
Then it is possible to construct a variation of light rays $\mathbf{f}:I\times \left[0,1\right]\rightarrow \overline{V}\subset M$ such that $\mathbf{f}\left(\tau,\cdot\right)\subset \gamma_{\tau}\in X\cap\mathcal{U}$ and $\mathbf{f}\left(\tau,1\right)=x$ for all $\tau\in I$. 
It is clear that for all $\tau\in I$ we have
\[
\Gamma'\left(\tau\right)\in T_{\gamma_{\tau}}X
\]
and using the definition of $\oplus^{V}$, then 
\[
\oplus^{V}_{\Gamma\left(\tau\right)}=\oplus^{V}_{\gamma_{\tau}}= \lim_{s\mapsto 1^{-}}T_{\gamma_{\tau}}S\left(\gamma_{\tau}\left(s\right)\right) = T_{\gamma_{\tau}}S\left(\gamma_{\tau}\left(1\right)\right) = T_{\gamma_{\tau}}S\left(\mathbf{f}\left(\tau,1\right)\right) = T_{\gamma_{\tau}}X
\]
and therefore, for all $\tau\in I$ 
\[
\Gamma'\left(\tau\right)\in \oplus^{V}_{\Gamma\left(\tau\right)} \, .
\]
This implies that the orbit $X^{+}\in \partial^{+}\Sigma_{V}$ of $\oplus^{V}$ going across $\gamma$ is just the set of light rays of the sky $X$ coming out of $V$. 
So, for any of such extendible space-time  $V$,  the $l$-boundary is made up of skies of points at the boundary of $V$. 

Let us denote by $\gamma_V = \gamma \cap V$ the segment of the light ray $\gamma$ contained in $V$.
Consider any $\gamma, \mu \in X^{+}\in \partial^{+}\Sigma_V$ and any $q\in I^{-}\left(\gamma_V,V\right)$.
Since $x\in I^{+}\left(q\right)$ then $\mu_V \cap I^{+}\left(q\right) \neq \varnothing$ and hence there is a timelike curve $\lambda:\left[0,1\right]\rightarrow M$ such that $\lambda\left(0\right)=q\in V$ and $\lambda\left(1\right)\in \mu_V \subset V$. 
But this implies that $\lambda\subset V$ because its endpoints are in a causally convex open set, therefore $q\in I^{-}\left(\mu_V,V\right)$. 
This shows that $I^{-}\left(\gamma_V,V\right) = I^{-}\left(\mu_V,V\right)$ for any $\gamma,\mu\in X^{+}$ and therefore there is a well defined map between the future GKP $c$-boundary  and the future $l$-boundary of $V$ given by:
\[
X^{+} \mapsto I^{-}\left(\gamma_V,V\right)
\]
because it is independent of the chosen light ray $\gamma \in X^{+}$

Since there are no imprisoned causal curves in $V$, every light ray $\gamma_V \subset V$ has endpoints in the boundary $\partial V \subset M$, it follows that
\[
\widetilde{\mathcal{U}}\subset \widetilde{\mathcal{N}} \subset \mathbb{P}\left(\mathcal{H}\right)
\]
is an open manifold with boundary and therefore
\[
\partial^{+}\widetilde{\mathcal{U}}\hookrightarrow \widetilde{\mathcal{N}}.
\]
is a homeomorphism onto its image.

We have proven above that any orbit $X^{+}$ of $\oplus^{V}$ is contained in the sky $X=S\left(x\right)$ where $x\in \partial V$, then the set of leaves in the foliation $\left(\mathcal{D}_{V}^{+}\right)^{\sim}$ of tangent spaces to the orbits coincide with the set of leaves in the foliation $\left(\mathcal{D}\right)^{\sim}$ of tangent spaces to the skies of points of $M$ restricted to $\partial^{+}\widetilde{\mathcal{U}}$. Thus using equation (\ref{boundary-chain}) we get:
\[
\left(\partial^{+}\Sigma_V\right)^{\sim} \simeq \partial^{+}\widetilde{\mathcal{U}} / \left(\mathcal{D}_{V}^{+}\right)^{\sim} = \partial^{+}\widetilde{\mathcal{U}} / \mathcal{D}^{\sim} \subset \widetilde{\mathcal{N}} / \mathcal{D}^{\sim} = \Sigma^{\sim} \, .
\]
Using now the inverse of the diffeomorphism $S^{\sim}: M \rightarrow \Sigma^{\sim}$ of Lemma \ref{conj20}, we obtain that  $\left(S^{\sim}\right)^{-1}\left(\partial^{+}\widetilde{\mathcal{U}} / \mathcal{D}^{\sim}\right)$ is contained in $\partial V$, then the topology of $\left(\partial^{+}\Sigma_V\right)^{\sim}\simeq\left(S^{\sim}\right)^{-1}\left(\partial^{+}\widetilde{\mathcal{U}} / \mathcal{D}^{\sim}\right)$, and therefore also of $\partial^{+}\Sigma_V$, is induced by the ambient manifold $M$.
Moreover, observe that $\left(S^{\sim}\right)^{-1}\left(\partial^{+}\widetilde{\mathcal{U}} / \mathcal{D}^{\sim}\right)$ is formed by all points in $\partial V$ accessible by a light ray. 

We consider now the case where no open segment of any light ray passing through $V$ is contained in $\partial V$, that is, we have the following definition:

\begin{definition} We will say that $p \in \partial V \subset M$ is light-transverse if any segment of light ray $\gamma:\left[a,b\right]\rightarrow M$ with $p \in \gamma$ and such that  $\gamma\left(a\right)\in V$ and $\gamma\left(b\right)\notin V$ satisfies that $\gamma \cap \partial V = \{ p \}$.   We will say that $V$ is light-transverse if every $p\in \partial V$ is light-transverse.
\end{definition}
This is clearly satisfied for $V=I^{+}\left(x\right)\cap I^{-}\left(y\right)$ such that $J^{+}\left(x\right)\cap J^{-}\left(y\right)$ is closed.  Notice that  if $M$ is a causally simple space-time then $J^{\pm}\left(x\right)$ is closed, then the previous set $V$ will be light-transverse.  
Then, it is easy to show that for any $\overline{p}\in\partial V$ accessible by light rays in $V$ there is a neighbourhood $W\subset \partial V$ such that any $\overline{q}\in W$ is accessible by light rays in $V$.

So, let us assume that there is a light ray $\gamma$ passing through a given $\overline{p}\in\partial V$. 
We can take a relatively compact, differentiable, space-like local hypersurface $C$ such that $\overline{p}\in C-\partial C$. 
If $\gamma$ is parametrized as the future--directed null geodesic verifying $\gamma\left(0\right)=\overline{p}$, then we can construct a non--zero differentiable null vector field $\widetilde{Z}\in \mathfrak{X}_C$ on $C$ such that $\widetilde{Z}_{\overline{p}} =  \gamma'\left(0\right)$. 
Under these conditions, we will apply the following result.

\begin{lemma}\label{extend-Z}
Let $\widetilde{C}$ be a  differentiable, local space-like hypersurface and $\widetilde{Z}\in \mathfrak{X}(\widetilde{C})$ a non-zero differentiable vector field defined on $\widetilde{C}$ and transverse to $\widetilde{C}$, then for any differentiable spacelike surface $C \subset \widetilde{C}$ such that $C$ is relatively compact in $\widetilde{C}$,    there exists $\epsilon>0$ such that  
\[
\begin{tabular}{ccl}
$F\colon C\times\left(-\epsilon , \epsilon \right)$ & $\rightarrow$ & $M$ \\
 $\left(p,s\right)$ & $\mapsto$ & $F\left(p,s\right)=\mathrm{exp}_p \left(s \widetilde{Z}_p\right)$
\end{tabular}
\]
is a diffeomorphism onto its image.
\end{lemma}

\begin{proof}
For every  $p\in \widetilde{C}$ there are a neighbourhood $U^{p}\subset \widetilde{C}$ and $\delta_p >0$ such that for all $x\in U^{p}$ the geodesic $\gamma_x\left(s\right)\equiv\mathrm{exp}_x\left(s\widetilde{Z}_x\right)$ is defined for all $s < \left|\delta_p\right|$ without conjugate points. 
Since $C$ is relatively compact in $\widetilde{C}$, there exists a finite subcovering $\left\{U^{p_i}\right\}$ of $C$.

Fixing $\delta=\mathrm{min}\left\{\delta_{p_{i}}\right\}$ then for all $p\in C$ the null geodesic $\gamma_p\left(s\right)$ is defined for $s< \left|\delta\right|$.
Then we can define
\[
\begin{tabular}{ccl}
$F\colon C\times\left(-\delta , \delta \right)$ & $\rightarrow$ & $M$ \\
$\left(p,s\right)$ & $\mapsto$ & $F\left(p,s\right)=\mathrm{exp}_p (s \widetilde{Z}_p ) \, ,$
\end{tabular}
\]
and if $q=F\left(p,s\right)=\gamma_p\left(s\right)$ then $Z_q\equiv\gamma'_p\left(s\right)$ is an extension of $\widetilde{Z}$ to the open neighbourhood of $C$ given by $\overline{W}=F\left(C\times \left(-\delta, \delta\right)\right)\subset M$.
By the locality of $C$, we can choose an orthonormal frame $\left\{\widetilde{E}_j\right\}$ on $C$ and propagate it to the whole $\overline{W}$ by parallel transport along every $\gamma_p$ for all $p\in C$.
For every $\left(p,0\right)\in C\times \left(-\delta , \delta\right)$ we have
\[
\begin{tabular}{l}
$dF_{\left(p,0\right)}\left(\left(\mathbf{0}_p,\left.\frac{\partial}{\partial s}\right|_{0}\right)\right) = \widetilde{Z}_p\in T_pM$ \\
\\
$dF_{\left(p,0\right)}\left(((\widetilde{E}_j)_p,\mathbf{0}_0)\right) = (\widetilde{E}_j )_p \in T_p M$
\end{tabular}
\]
where $\frac{\partial}{\partial s}$ is the tangent vector field of the curves $\alpha_q\left(s\right)=\left(q,s\right)\in C\times \left(-\delta,\delta\right)$. 
Since $dF_{\left(p,0\right)}$ maps a basis of $T_{\left(p,0\right)}\left(C\times \mathbb{R}\right)\approx T_pC \times T_0 \mathbb{R}$ into a basis of $T_pM$, then it is an isomorphism and hence $F$ is a local diffeomorphism.
So, there exists a neighbourhood $H^p \times \left(-\epsilon_p , \epsilon_p\right)$ of $\left(p,0\right)\in C\times \left(-\delta,\delta\right)$ with $0<\epsilon_p <\delta$ such that the restriction of $F$ is a diffeomorphism. 
Again, since $C$ is relatively compact, then from the covering $\left\{H^p\right\}$ we can extract a finite subcovering $\left\{H^k\right\}$ of $C$, then taking $\epsilon = \mathrm{min}\left\{\epsilon_k\right\}$ we have 
\[
C\times \left(-\epsilon,\epsilon\right)=\bigcup_k H^k \times \left(-\epsilon,\epsilon\right)
\]
Calling $W=F\left(C\times \left(-\epsilon,\epsilon\right)\right)$ then for any $\left(p,s\right)\in C\times \left(-\epsilon,\epsilon\right)$, the map $F:C\times \left(-\epsilon,\epsilon\right)\rightarrow W$ is a local diffeomorphism.
By construction, this restriction of $F$ is surjective, and since there are not conjugated points in the null geodesics $\gamma_q$, then we get the injectivity. 
Therefore we conclude that $F:C\times \left(-\epsilon,\epsilon\right)\rightarrow W$ is a global diffeomorphism.
\end{proof}
\vspace{3mm}

If we apply now Lemma \ref{extend-Z} to the proposed hypersurface $C$, then the image of the map $F$ is an open neighbourhood of $\overline{p}\in M$. 
We can take a nested sequence $\left\{C_n\right\}\subset C$ of neighbourhoods of $\overline{p}$ in $C$ converging to $\left\{\overline{p}\right\}$ and restrict $F$ to $C_n\times\left(-\epsilon , \epsilon \right)$. 
Let us assume that for every $C_n$ there exists a null geodesic segment $\gamma_n=F\left(q_n,\left(0 , \epsilon \right)\right)$  fully contained in $V$, then for any $0<s< \epsilon $ the sequence $F\left(q_n,s\right)\mapsto \gamma\left(s\right)$ as $n$ increases.
Hence $\gamma\left(\left(0, \epsilon\right)\right)\subset \partial V$ since $\gamma\left(\left(0,\epsilon\right)\right)\cap V=\varnothing$, therefore $\left.\gamma\right|_{\left(0,\epsilon\right)}$ is contained in $\partial V$ contradicting that there is no segment of a light ray contained in $\partial V$.

On the other hand, if for every $C_n$ there is a null geodesic segment $\gamma_n=F\left(q_n,\left(-\epsilon , 0\right)\right)$ without points in $V$, then as done before, we have that $\gamma\left(\left(-\epsilon,0\right)\right)\subset \partial V$ but this contradicts that $\gamma\left(\left(-\epsilon,0\right)\right)\subset V$. 

Therefore, there exist $C_k\subset C$ such that for all $q\in C_k$ the null geodesic segment $\gamma_q = F\left(q,\cdot\right)$ has endpoints $\gamma_q\left(s_1\right)\in V$ and $\gamma_q\left(s_2\right)\in M-V$ with $-\epsilon < s_1 < s_2 < \epsilon$.
Since $\partial V$ is a topological hypersurface then $B=F\left(C_k,\left(-\epsilon , \epsilon \right)\right)\cap \partial V$ is an open set of $\partial V$ such that all points in $B$ are accessible by future--directed null geodesic.  Hence we conclude that the set of light-transverse points in $\partial V$ is an open set relative to $\partial V$ with the induced topology from $M$.

Then we may consider the open subset $\partial V_r$ of the future $l$-boundary $\partial^{+}\Sigma_V$ consisting of light-transverse accesible by null geodesic points in $\partial V$.
It is also known that the future $c$-boundary of $V$ is also topologically equivalent to $\partial V\subset M$, so the future $l$-boundary is equivalent to the future $c$-boundary in the set $\partial V_r$.
Thus we have proved:

\begin{proposition}  Let $V\subset M$ be a light-transverse, globally hyperbolic, causally convex, convex normal neighbourhood of $M$.  Then the $l$-boundary, $c$-boundary and topological boundary $\partial V$ of $V$ coincide in the set of light-transverse points in $\partial V$ which are accessible by null geodesics in $V$.
\end{proposition}

The previous procedure can be carried out for more general space-times $V$.  The only condition needed is light-transversality at points in the boundary, meaning by that that any null geodesic $\gamma_q$ defined by the diffeomorphism $F$ intersects $\partial V$ ``transversally'' even if $\partial V$ is not smooth (that is, crossing $\partial V$ and not remaining in $\partial V$ for any interval of the parameter of $\gamma_q$).   Clearly, if $\partial V$ is a smooth submanifold this notion becomes just ordinary  transversality.

Now, how can we deal with a general case in order to calculate points in the $l$-boundary when there is not any larger space-time  containing $M$? 
We can use the previous calculations. 
Consider any light ray $\gamma\in \mathcal{N}$, then we can parametrize an inextensible future--directed segment of it by $\gamma:\left[0,b\right)\rightarrow M$. 
We can cover this segment by means of a countable collection $\left\{V_n\right\}$ formed by relatively compact globally hyperbolic, causally convex and convex normal neighbourhoods $V_n$.
Without any lack of generality, we can assume that $V_n \cap V_k \neq \varnothing$ if and only if $n = k\pm 1$ and $n$ increases when $\gamma\left(s\right)$ moves to the future.
If we denote by $x_n\in \partial V_n$ the future endpoint of $\gamma\cap V_n$, then the orbit of $\oplus^{V_n}$ passing through $\gamma$ is $X_n\cap \mathcal{U}_n \subset \mathcal{N}$, or in other words, it is defined by $X_n \in \Sigma$. 
In this way, the orbit $X^{+}\in \partial ^{+} \Sigma$ of $\oplus:\mathcal{N}\rightarrow \mathbb{P}\left(\mathcal{H}\right)$ can be constructed by the limit in $\mathcal{N}$ of the sequence $\left\{X_n\right\}$ if such limit exists, something that automatically happens in dimension three as we saw in Section~\ref{sec:Low-boundary}.

We may summarize the previous discussion in the following Proposition.

\begin{proposition}  Let $(M,\mathcal{C})$ be a strongly causal sky-separating conformal space-time such that the future limit distribution $\oplus$ exists and such that there is an extension of the conformal structure to the future $l$-boundary $\partial^+ \Sigma$ of $M$ (similarly for the past $l$-boundary $\partial^-\Sigma$).  The future $l$-boundary is equivalent to the future $c$-boundary in the set of light-transverse points in $\partial M = \overline{M} \backslash M$ accesible by future-directed null geodesics  in $\overline{M}$.
\end{proposition}



\section[Examples]{Some examples}\label{sec:examples}

\setcounter{nseccion}{4}

In the present section, we offer some examples in which the previously studied structures will be discussed explicitly.
Although we will focus on $3$--dimensional space-times, we will also deal with $4$--dimensional Minkowski space-time  that will turn out to be useful in the study of two embedded $3$--dimensional examples: Minkowski and de Sitter space-times. 
In these two examples, we will proceed restricting them from the $4$--dimensional Minkowski example as section \ref{sec:embedded} suggests.


\subsection{Embedded spaces of light rays}\label{sec:embedded}

Now, we will deal with some particular cases of embedded space-times.
Let $\overline{M}$ be a $\left(m+1\right)$--dimensional, strongly causal and null pseudo--convex space-time  with metric $\overline{\mathbf{g}}$ where $m\geq 3$. 
We will denote overlined its structures $\overline{\mathcal{N}}$, $\overline{\mathcal{H}}$, etc.
Consider $M \subset \overline{M}$ an embedded $m$--dimensional, strongly causal and null pseudo--convex space-time  equipped with the metric $\mathbf{g}=\left.\overline{\mathbf{g}}\right|_M$ such that any maximal null geodesic in $M$ is a maximal null geodesic in $\overline{M}$. 
Since $M$ is embedded in $\overline{M}$, then trivially $TM$ is embedded in $T\overline{M}$. 

Given a  globally hyperbolic, causally convex and convex normal  open set $\overline{V}\subset \overline{M}$ such that $\overline{C}\subset \overline{V}$ is a smooth space-like Cauchy surface, then clearly $V=\overline{V}\cap M$ is causally convex and contained in a convex normal neighbourhood. 
Moreover, if $\lambda\subset V$ is an inextensible time-like curve, since $\lambda\subset \overline{V}$ then $\lambda$ intersects exactly once to $\overline{C}$, hence the intersection point must be in $C=\overline{C}\cap M$ and therefore $C\subset V$ is a smooth space-like Cauchy surface in $V$. 
This implies that $V$ is also a globally hyperbolic  open set in $M$.

Since the inclusion $TV\hookrightarrow T\overline{V}$ is an embedding, its restriction  $\mathbb{N}\left(C\right)\hookrightarrow\mathbb{N}\left(\overline{C}\right)$ is also an embedding.
Given a fixed timelike vector field $Z\in \mathfrak{X}\left(V\right)$, since $\overline{V}$ is an arbitrary   globally hyperbolic, causally convex and convex normal  open set, without any lack of generality, we can choose any time-like extension $\overline{Z}\in \mathfrak{X}\left(\overline{V}\right)$ of $Z$, that is $Z=\left.\overline{Z}\right|_{V}$. 
For all $v\in \mathbb{N}\left(C\right)\subset \mathbb{N}\left(\overline{C}\right)$ we have 
\[
\mathbf{g}\left(v,Z\right)=\mathbf{g}\left(v,\overline{Z}\right)
\]
Then the map, 
\[
\Omega^{Z}\left(C\right)=\left\{v\in \mathbb{N}\left(C\right): \mathbf{g}\left(v,Z\right)=-1 \right\} \hookrightarrow \Omega^{\overline{Z}}\left(\overline{C}\right)=\left\{v\in \mathbb{N}\left(\overline{C}\right): \mathbf{g}\left(v,\overline{Z}\right)=-1 \right\}
\]
is an embedding.
Again, since  $\mathcal{U}\simeq \Omega^{Z}\left(C\right)$ and $\overline{\mathcal{U}}\simeq\Omega^{\overline{Z}}\left(\overline{C}\right)$, then we have that the inclusion 
\[
\mathcal{N}\supset\mathcal{U} \hookrightarrow \overline{\mathcal{U}}\subset\overline{\mathcal{N}} \, ,
\]
is an embedding. 
Since $\mathcal{N} \hookrightarrow \overline{\mathcal{N}}$ is an inclusion, then it is injective and thus a global embedding.
Therefore also 
\[
T\mathcal{N} \hookrightarrow T\overline{\mathcal{N}}
\]
is another global embedding.

Given a point $x\in M\subset \overline{M}$, its sky $X\in \Sigma$ is the set of all light rays contained in $\mathcal{N}$ passing through $x$, but since every light ray in $\mathcal{N}$ is a light ray in $\overline{\mathcal{N}}$, then calling $\overline{X}\in \overline{\Sigma}$ the sky of $x$ relative to $\overline{\mathcal{N}}$ we have
\[
X = \overline{X}\cap \mathcal{N} \, .
\]

Since the metric in $M$ is just the restriction to $TM$ of the metric in $\overline{M}$, then the contact structure $\mathcal{H}$ of $\mathcal{N}$ is the restriction of the contact structure $\overline{\mathcal{H}}$ of $\overline{\mathcal{N}}$ to the tangent bundle $T\mathcal{N}$, that is
\[
\mathcal{H}_{\gamma} = \overline{\mathcal{H}}_{\gamma} \cap T_{\gamma}\mathcal{N}
\]
for all $\gamma\in \mathcal{N}$.
So, for any $\gamma\in X \subset\mathcal{N}$, it is now clear that 
\[
T_{\gamma}X = T_{\gamma}\overline{X}\cap T_{\gamma}\mathcal{N}=T_{\gamma}\overline{X}\cap \mathcal{H}_{\gamma}
\]
due to $T_{\gamma}X\subset \mathcal{H}_{\gamma}$. 
For a regular parametrization $\gamma:\left(a,b\right)\rightarrow M$, we can write 
\[
T_{\gamma}S\left(\gamma\left(s\right)\right) = T_{\gamma}\overline{S\left(\gamma\left(s\right)\right)}\cap  \mathcal{H}_{\gamma}
\]
and hence, the future limit distribution $\oplus$ is given as:
\[
\oplus_{\gamma}=\lim_{s\mapsto b^{-}} T_{\gamma}S\left(\gamma\left(s\right)\right) = \lim_{s\mapsto b^{-}} T_{\gamma}\overline{S\left(\gamma\left(s\right)\right)}\cap  \mathcal{H}_{\gamma} = \overline{\oplus}_{\gamma} \cap  \mathcal{H}_{\gamma} \, .
\]

If the distribution defined by $\overline{\oplus}$ in $\overline{\mathcal{N}}$ is integrable, then the orbits of $\oplus$ become the orbits of $\overline{\oplus}$ restricted to $\mathcal{N}$, that is
\[
X^{+} = \overline{X}^{+} \cap \mathcal{N} \, .
\] 

After the previous considerations, we can use the contents of the current section to study $3$--dimensional Minkowski and de Sitter space-times as embedded in a $4$--dimensional Minkowski space-time.


\subsection{$4$--dimensional Minkowski space-time}\label{sec:Mink-4}
Consider the 4-dimensional Minkowski space-time  given by $\mathbb{M}^{4}=\left(\mathbb{R}^{4}, \mathbf{g}\right)$ where the metric is given by $\mathbf{g}=-dt\otimes dt + dx\otimes dx + dy\otimes dy + dz\otimes dz$ in the standard coordinate system $\varphi=\left(t,x,y,z\right)$.
We will use the notation $\overline{\mathcal{N}}$, $\overline{\mathcal{H}}$, etc., for the structures related to $\mathbb{M}^{4}$. 

It is known that the hypersurface $\overline{C}\equiv\left\{t=0\right\}$ is a global Cauchy surface then $\overline{\mathcal{N}}$ is diffeomorphic to $\overline{C}\times \mathbb{S}^{2}$ \cite[Sect. 4]{Ch10}. 
We can describe points at the sphere $\mathbb{S}^{2}$ using spherical coordinates $\theta$, $\phi$. Then, we can use $\psi=\left(x,y,z,\theta,\phi\right)$ as a system of coordinates in $\overline{\mathcal{N}}$, where $\psi^{-1}\left(x_0,y_0,z_0,\theta_0,\phi_0\right)=\gamma\in \overline{\mathcal{N}}$ corresponds to the light ray given by
\[
\gamma\left(s\right)=\left(s \hspace{1mm},\hspace{2mm} x_0 + s\cdot \cos \theta_0 \sin \phi_0 \hspace{1mm}, \hspace{2mm} y_0 + s\cdot \sin \theta_0 \sin \phi_0 \hspace{1mm}, \hspace{2mm} z_0 + s\cdot \cos \phi_0  \right)
\]
with $s\in \mathbb{R}$.  

In general, it is possible to calculate the contact hyperplane at $\gamma\in \overline{\mathcal{N}}$ as the vector subspace in $T_{\gamma}\overline{\mathcal{N}}$ generated by tangent spaces to the skies at two different non--conjugate points in $\gamma$, or in other words, if $\gamma\left(s_1\right)$ and $\gamma\left(s_2\right)$ are not conjugate along $\gamma$ then $T_{\gamma}S\left(\gamma\left(s_1\right)\right)\cap T_{\gamma}S\left(\gamma\left(s_2\right)\right)=\left\{ \mathbf{0} \right\}$ and by  dimension counting we see that 
\[
\overline{\mathcal{H}}_{\gamma} = T_{\gamma}S\left(\gamma\left(s_1\right)\right)\oplus T_{\gamma}S\left(\gamma\left(s_2\right)\right) \, .
\]
In case of Minkowski space-time there are no conjugate points along any geodesics, so we will use for this purpose the points $\gamma\left(0\right)$ and any $\gamma\left(s\right)$.
Thus fixed $s$, for any $\left(\theta, \phi\right)$, the curve 
\[
\mu_{\left(\theta, \phi\right)}\left(\tau\right) = \gamma\left(s\right) + \tau \left(1\hspace{1mm},\hspace{2mm} \cos \theta \sin \phi \hspace{1mm}, \hspace{2mm} \sin \theta \sin \phi \hspace{1mm}, \hspace{2mm} \cos \phi  \right) \, ,
\]
describes a null geodesic passing by $\gamma\left(s\right)$ that cut $\overline{C}$ at $\tau=-s$. 
So, the sky of $\gamma\left(s\right)$ can be written in coordinates by
\[
\psi\left(S\left(\gamma\left(s\right)\right)\right) \equiv \left\{
\begin{array}{l}
x\left(\theta, \phi\right) =  x_0 +s\left( \cos \theta_0 \sin \phi_0 - \cos \theta \sin \phi  \right) \, , \\
y\left(\theta, \phi\right) =  y_0 +s\left( \sin \theta_0 \sin \phi_0 - \sin \theta \sin \phi  \right) \, , \\
z\left(\theta, \phi\right) =  z_0 +s\left( \cos \phi_0  - \cos \phi  \right) \, , \\
\theta\left(\theta, \phi\right) =  \theta \, , \\
\phi\left(\theta, \phi\right) =  \phi \, ,
\end{array}
\right. 
\]
and the derivatives of these expressions with respect to $\theta$ and $\phi$ at $\left(\theta, \phi\right)=\left(\theta_0, \phi_0\right)$ give us the generators of the tangent space of the sky $S\left(\gamma\left(s\right)\right)$ at $\gamma$, so
\begin{align*}
T_{\gamma}S\left(\gamma\left(s\right)\right) &= \mathrm{span}\left\{ \textstyle{   s\left( \sin\theta_0 \sin \phi_0 \left( \frac{\partial}{\partial x} \right)_{\gamma} - \cos\theta_0 \sin \phi_0 \left( \frac{\partial}{\partial y} \right)_{\gamma} \right) + \left( \frac{\partial}{\partial \theta} \right)_{\gamma}   } , \right. \\
& \left. \textstyle{   s\left( -\cos\theta_0 \cos \phi_0 \left( \frac{\partial}{\partial x} \right)_{\gamma} - \sin\theta_0 \cos \phi_0 \left( \frac{\partial}{\partial y} \right)_{\gamma}  +\sin \phi_0 \left( \frac{\partial}{\partial z} \right)_{\gamma} \right) + \left( \frac{\partial}{\partial \phi} \right)_{\gamma}    }   \right\}
\end{align*}
and trivially
\[
T_{\gamma}S\left(\gamma\left(0\right)\right) = \mathrm{span}\left\{  \textstyle{ \left( \frac{\partial}{\partial \theta} \right)_{\gamma} ,  \left( \frac{\partial}{\partial \phi} \right)_{\gamma}   }    \right\} \, .
\]
Therefore the contact hyperplane at $\gamma$ is
\begin{align*}
\overline{\mathcal{H}}_{\gamma} &= \mathrm{span}\left\{ \textstyle{   \left( \frac{\partial}{\partial \theta} \right)_{\gamma}, \left( \frac{\partial}{\partial \phi} \right)_{\gamma}, \sin\theta_0  \left( \frac{\partial}{\partial x} \right)_{\gamma} - \cos\theta_0  \left( \frac{\partial}{\partial y} \right)_{\gamma}   }   , \right. \\
& \left. \textstyle{   \cos\theta_0 \cos \phi_0 \left( \frac{\partial}{\partial x} \right)_{\gamma} + \sin\theta_0 \cos \phi_0 \left( \frac{\partial}{\partial y} \right)_{\gamma}  - \sin \phi_0 \left( \frac{\partial}{\partial z} \right)_{\gamma}  } \right\}
\end{align*}
and a contact form is given by:
\[
\overline{\alpha} = \cos \theta \sin \phi \cdot dx + \sin \theta \sin \phi \cdot dy + \cos \phi \cdot dz \, .
\]

For this space-time it is easy to calculate the limit distributions $\overline\oplus$ and $\overline\ominus$.
We will proceed only for $\overline\oplus$ because the case of $\overline\ominus$ is analogous.
Using the definition (\ref{boundary-field}), we have 
\begin{align*}
\overline{\oplus}_{\gamma} &= \lim_{s\mapsto +\infty}T_{\gamma}S\left(\gamma\left(s\right)\right)= \\
&= \mathrm{span}\left\{ \textstyle{ \sin\theta_0 \sin \phi_0 \left( \frac{\partial}{\partial x} \right)_{\gamma} - \cos\theta_0 \sin \phi_0 \left( \frac{\partial}{\partial y} \right)_{\gamma}   }   \right. , \\
& \left. \textstyle{ -\cos\theta_0 \cos \phi_0 \left( \frac{\partial}{\partial x} \right)_{\gamma} - \sin\theta_0 \cos \phi_0 \left( \frac{\partial}{\partial y} \right)_{\gamma}  +\sin \phi_0 \left( \frac{\partial}{\partial z} \right)_{\gamma} } \right\} \, ,
\end{align*}
and therefore $\overline{\oplus}$ defines a integrable distribution whose partial differential equations are: 
\[
\left\{
\begin{array}{l}
\displaystyle{\frac{\partial x }{\partial \alpha}}\left(\alpha, \beta\right) = \sin \theta \sin \phi \vspace{1mm} \\
\displaystyle{\frac{\partial y }{\partial \alpha}}\left(\alpha, \beta\right) = -\cos \theta \sin \phi  \vspace{1mm}\\
\displaystyle{\frac{\partial z }{\partial \alpha}}\left(\alpha, \beta\right) = 0  \vspace{1mm}\\
\displaystyle{\frac{\partial \theta }{\partial \alpha}}\left(\alpha, \beta\right) = 0  \vspace{1mm}\\
\displaystyle{\frac{\partial \phi }{\partial \alpha}}\left(\alpha, \beta\right) = 0 
\end{array}
\right. 
\hspace{5mm}
\left\{
\begin{array}{l}
\displaystyle{\frac{\partial x }{\partial \beta}}\left(\alpha, \beta\right) = -\cos \theta \cos \phi \vspace{1mm} \\
\displaystyle{\frac{\partial y }{\partial \beta}}\left(\alpha, \beta\right) = -\sin \theta \cos \phi \vspace{1mm} \\
\displaystyle{\frac{\partial z }{\partial \beta}}\left(\alpha, \beta\right) = \sin \phi \vspace{1mm} \\
\displaystyle{\frac{\partial \theta }{\partial \beta}}\left(\alpha, \beta\right) = 0 \vspace{1mm} \\
\displaystyle{\frac{\partial \phi }{\partial \beta}}\left(\alpha, \beta\right) = 0 
\end{array}
\right. 
\]
and its solution with initial values $\left(x_0,y_0,z_0, \theta_0 , \phi_0\right)$, is given by:
\begin{equation}\label{solution-orbit-boundary}
\left\{
\begin{array}{l}
x\left(\alpha, \beta\right) =  x_0 +\alpha \sin \theta_0 \sin \phi_0 - \beta \cos \theta_0 \cos \phi_0 \\
y\left(\alpha, \beta\right) =  y_0 -\alpha \cos \theta_0 \sin \phi_0 - \beta \sin \theta_0 \cos \phi_0 \\
z\left(\alpha, \beta\right) =  z_0 +\beta \sin \phi_0 \\
\theta\left(\alpha, \beta\right) =  \theta_0 \\
\phi\left(\alpha, \beta\right) =  \phi_0
\end{array}
\right. 
\end{equation}

This solution corresponds to the $2$--plane 
\begin{equation}\label{plane-orbit-boundary}
\cos \theta_0 \sin \phi_0 \cdot \left(x-x_0\right) + \sin \theta_0 \sin \phi_0 \cdot \left(y-y_0\right) + \cos \phi_0 \cdot \left(z-z_0\right) = 0 \, ,
\end{equation}
in the Cauchy surface $\overline{C}$ and it defines the orbit $\overline{X}_{\gamma}^{+}$ of $\overline\oplus$ passing through $\gamma$. 
The image in $\mathbb{M}^4$ of all the light rays in $\overline{X}_{\gamma}^{+}$ is precisely the $3$--plane in $\mathbb{M}^{4}$ given by 
\[
\cos \theta_0 \sin \phi_0 \cdot \left(x-x_0\right) + \sin \theta_0 \sin \phi_0 \cdot \left(y-y_0\right) + \cos \phi_0 \cdot \left(z-z_0\right) - t = 0
\]
and it is easy to show, using straightforward calculations, that any light ray $\mu\in \overline{X}_{\gamma}^{+}$ in the same orbit of $\overline\oplus$ than $\gamma$ determines the TIP
\[
I^{-}\left(\mu\right) = I^{-}\left(\gamma\right) = \left\{ t < \cos \theta_0 \sin \phi_0 \cdot \left(x-x_0\right) + \sin \theta_0 \sin \phi_0 \cdot \left(y-y_0\right) + \cos \phi_0 \cdot \left(z-z_0\right)    \right\} \, ,
\]
so the future $l$-boundary coincides with $c$-boundary except for the TIP $I^{-}\left(\lambda\right)=\mathbb{M}^4$ defined by any time-like geodesic $\lambda$, because it can not be defined by light rays. 

Moreover \cite[Thm. 4.16]{FHS11} ensures that, for this space-time, the $c$--boundary is the same as the conformal boundary.  
The $l$-boundary corresponds to the set of all orbits of $\overline\oplus$, that is, all $2$-planes (\ref{plane-orbit-boundary}). 
Observe that the map 
\begin{equation}\label{map-boundary}
\begin{tabular}{rcl}
$\mathbb{R}^{3}\times\mathbb{S}^{2} \simeq \overline{\mathcal{N}}$ & $\rightarrow$ & $\partial^{+}\overline{\Sigma} \simeq \mathbb{R}^{1}\times\mathbb{S}^{2}$\\
$\gamma$ & $\mapsto$ & $\overline{X}_{\gamma}^{+}$
\end{tabular}
\end{equation}
such that every light ray $\gamma\in \mathcal{N}$ is mapped to the point of the $l$-boundary corresponding to the orbit of $\overline\oplus$ passing through $\gamma$ can be written in coordinates by
\[
\left(x,y,z, \theta , \phi\right) \mapsto \left(\cos \theta \sin \phi \cdot x + \sin \theta \sin \phi \cdot y + \cos \phi \cdot z, \theta , \phi\right) \, ,
\]
therefore the future  $l$-boundary is $\partial^{+}\overline{\Sigma} \simeq \mathbb{R}^{1}\times\mathbb{S}^{2}$.


\subsection{$3$--dimensional Minkowski space-time}\label{sec:Mink-3}

Let us proceed now with $3$--dimensional Minkowski space-time  given by $\mathbb{M}^{3}=\left(\mathbb{R}^{3}, \mathbf{g}\right)$ with metric $\mathbf{g}=-dt\otimes dt + dx\otimes dx + dy\otimes dy$ in coordinates $\varphi=\left(t,x,y\right)$. 
We will use the notation $\mathcal{N}$, $\mathcal{H}$, etc., for the structures related to $\mathbb{M}^{3}$. 

It is possible to see $\mathbb{M}^{3}$ as the restriction of $\mathbb{M}^{4}$ to its hyperplane $z=0$. 
So, in order to obtain the description of the space of light rays of $\mathbb{M}^3$, we can restrict the results obtained in section~\ref{sec:Mink-4} to $z=0$ and therefore, with $\phi = \pi / 2$. 

Then, $C\equiv\left\{t=0\right\}$ is still a Cauchy surface and $\mathcal{N}\simeq C\times \mathbb{S}^{1}$ and we can use $\psi=\left(x,y,\theta\right)$ as a system of coordinates in $\mathcal{N}$, where $\psi^{-1}\left(x_0,y_0,\theta_0\right)=\gamma\in \mathcal{N}$ describes the light ray given by
\[
\gamma\left(s\right)=\left(s \hspace{1mm},\hspace{2mm} x_0 + s\cdot \cos \theta_0  \hspace{1mm}, \hspace{2mm} y_0 + s\cdot \sin \theta_0  \right)
\]
with $s\in \mathbb{R}$.

So, the tangent space of the skies $S\left(\gamma\left(s\right)\right)$ and $S\left(\gamma\left(0\right)\right)$ at $\gamma$ can be written as
\begin{equation}\label{tangent-Mink-sky}
T_{\gamma}S\left(\gamma\left(s\right)\right) = \mathrm{span}\left\{   \textstyle{   s\left( \sin\theta_0 \left( \frac{\partial}{\partial x} \right)_{\gamma} - \cos\theta_0 \left( \frac{\partial}{\partial y} \right)_{\gamma} \right) + \left( \frac{\partial}{\partial \theta} \right)_{\gamma}   }    \right\}
\end{equation}
and 
\[
T_{\gamma}S\left(\gamma\left(0\right)\right) = \mathrm{span}\left\{  \textstyle{  \left( \frac{\partial}{\partial \theta} \right)_{\gamma}    }    \right\} \, .
\]
Therefore the contact hyperplane at $\gamma$ is
\[
\mathcal{H}_{\gamma} = \mathrm{span}\left\{ \textstyle{    \sin\theta_0 \left( \frac{\partial}{\partial x} \right)_{\gamma} - \cos\theta_0 \left( \frac{\partial}{\partial y} \right)_{\gamma} , \left( \frac{\partial}{\partial \theta} \right)_{\gamma}    }     \right\}
\]
and any contact form will be proportional to 
\[
\alpha = \cos \theta \cdot dx + \sin \theta \cdot dy  \, .
\]

Using (\ref{tangent-Mink-sky}) it is possible to calculate easily the point in the $l$-boundary passing by $\gamma$, then
\[
\oplus_{\gamma}=\lim_{s\mapsto +\infty}T_{\gamma}S\left(\gamma\left(s\right)\right) = \mathrm{span}\left\{    \textstyle{    \sin\theta_0 \left( \frac{\partial}{\partial x} \right)_{\gamma} - \cos\theta_0 \left( \frac{\partial}{\partial y} \right)_{\gamma}    }    \right\}
\]
and therefore we can obtain the integral curve $c\left(\tau\right)=\left(x\left(\tau\right),y\left(\tau\right),\theta\left(\tau\right)\right)$ defining the orbit $X^{+}_{\gamma}\subset \mathcal{N}$ of $\oplus$ containing $\gamma$ solving the initial value problem
\[
\left\{
\begin{array}{l}
x'\left(\tau\right) =  \sin \theta \\
y'\left(\tau\right) =  -\cos \theta \\
\theta '\left(\tau\right) =  0 \\
c\left(0\right) = \left(x_0,y_0,\theta_0 \right)
\end{array}
\right. 
\]
Its solution is $c\left(\tau\right)=\left(x_0 + \tau \sin \theta_0 \, , \, y_0 - \tau \cos \theta_0 \, , \, \theta_0\right)$ and corresponds to the family of null geodesics with tangent vector $v=\left(1, \cos \theta_0, \sin \theta_0 \right)$ and initial value in the straight line contained in $C$ given by 
\[
\left\{
\begin{array}{l}
\cos \theta_0 \left(x-x_0\right) + \sin \theta_0 \left(y-y_0\right) =0 \\
t= 0
\end{array}
\right. .
\] 
 
Again, by straightforward calculations, it is possible to show that given $\mu_1 , \mu_2 \in X^{+}_{\gamma}$ then $I^{-}\left(\mu_1\right) = I^{-}\left(\mu_2\right)$, therefore any light ray in $X^{+}_{\gamma}$ defines the same TIP
\[
I^{-}\left(\gamma\right) = \left\{ \left(t,x,y\right)\in \mathbb{M}^3: t< \cos \theta_0 \left(x-x_0\right) + \sin \theta_0 \left(y-y_0\right)  \right\} .
\]
then, again the future $l$-boundary coincides with the future part of the $c$-boundary accessible by light rays.

In an analogous way, the orbit $X^{-}_{\gamma}$ of $\ominus$ verifies $X^{-}_{\gamma} = X^{+}_{\gamma}$ and thus it corresponds to the TIF $I^{+}\left(\gamma\right)$.

The restriction of the map (\ref{map-boundary}) to $\mathcal{N}\simeq\mathbb{R}^{2}\times\mathbb{S}^{1}$ results
\[
\begin{tabular}{rcl}
$\mathbb{R}^{2}\times\mathbb{S}^{1}\simeq \mathcal{N}$ & $\rightarrow$ & $\partial^{+}\Sigma \simeq \mathbb{R}^{1}\times\mathbb{S}^{1}$\\
$\gamma$ & $\mapsto$ & $X^{+}_{\gamma}$
\end{tabular}
\]
that, in coordinates, can be written by
\[
\left(x,y,\theta \right) \mapsto \left(\cos \theta \cdot x + \sin \theta \cdot y , \theta \right)
\]
therefore, $\partial^{+}\Sigma \simeq \mathbb{R}^{1}\times\mathbb{S}^{1}$.

We can use the previous calculations to describe a globally hyperbolic block embedded in $\mathbb{M}^3$. 
Let us call $M_{*}=\left\{ \left(t,x,y\right)\in \mathbb{M}^3:t>-1 \right\}$ with the same metric $\mathbf{g}$ restricted to $M_{*}$, and denote by $\mathcal{N}_{*}$,  $\mathcal{H}_{*}$, etc., the corresponding structures for $M_{*}$. 
Since $M_{*}\subset \mathbb{M}^3$ is open and they share the same Cauchy surface $C\equiv\left\{t=0\right\}$, then trivially $\mathcal{N}_{*} \simeq \mathcal{N}$ and $\mathcal{H}_{*} \simeq \mathcal{H}$. 
To calculate $\ominus_{*}$, we can consider the limit of the expression (\ref{tangent-Mink-sky}) when $s$ tends to $-1$, then 
\[
\left(\ominus_{*}\right)_{\gamma}=\lim_{s\mapsto -1}T_{\gamma}S\left(\gamma\left(s\right)\right) = \mathrm{span}\left\{   \textstyle{   -\sin\theta_0 \left( \frac{\partial}{\partial x} \right)_{\gamma} + \cos\theta_0 \left( \frac{\partial}{\partial y} \right)_{\gamma} + \left( \frac{\partial}{\partial \theta} \right)_{\gamma}   }   \right\}
\]
Thus, the orbit $X^{-}_{\gamma}\subset \mathcal{N}_{*}$ of $\ominus_{*}$ passing by $\gamma$ is the solution $c\left(\tau\right)=\left( x\left(\tau\right), y\left(\tau\right), \theta\left(\tau\right) \right)$ of 
\[
\left\{
\begin{array}{l}
x'\left(\tau\right) =  -\sin \theta  \\
y'\left(\tau\right) =  \cos \theta  \\
\theta '\left(\tau\right) =  1  \\
c\left(0\right) = \left(x_0,y_0,\theta_0 \right)
\end{array}
\right. 
\]
and it is given by $c\left(\tau\right)=\left(x_0 + \cos \left(\tau+ \theta_0\right) \, , \, y_0 +\sin \theta_0\left(\tau+ \theta_0\right) \, , \, \tau + \theta_0\right)$. 
The light ray in $X^{-}_{\gamma}$ defined by $c\left(\tau\right)$ can be parametrized (as a null geodesic) by
\[
\gamma_{\tau}\left(s\right)=\left( s \, , \, x\left(\tau\right)+s \cos \theta\left(\tau\right) \, , \, y\left(\tau\right)+s \sin \theta\left(\tau\right)  \right) =
\]
\[
=\left( s \, , \, x_0 + \left(s+1\right)\cos \left(\tau+ \theta_0\right) \, , \, y_0 + \left(s+1\right)\sin \left(\tau+ \theta_0\right) \right) \, ,
\]
verifying $\lim_{s\mapsto -1}\gamma_{\tau}\left(s\right)=\left(-1,x_0,y_0\right)$ for all $\tau$. 
This clearly shows that $X^{-}_{\gamma}\subset \mathcal{N}_{*}$ can be identified with $S\left(\left(-1,x_0,y_0\right)\right)\subset \mathcal{N}$ and therefore the past $l$-boundary completed space $M_{*}\cup \partial^{-}\Sigma_{*}$ can be identified diffeomorphically with $\left\{ \left(t,x,y\right)\in \mathbb{M}^3:t\geq -1 \right\}$.  


\subsection{$3$--dimensional de Sitter space-time }
Using the notation of section \ref{sec:Mink-4}, we can define the \emph{de Sitter space-time } $S_{1}^{3}$ as the set in $\mathbb{M}^4$ verifying 
\begin{equation} \label{Secuacion}
-t^{2}+x^{2}+y^{2}+z^{2}=1 \, .
\end{equation}
We will denote the structures related to $S_{1}^{3}$ by $\mathcal{N}_S$,  $\mathcal{H}_S$, etc.
Because of \cite[Prop. 4.28]{On83}, light rays in $\mathcal{N}_S$ are straight lines in $\mathbb{M}^{4}$ contained in $S_{1}^{3}$, that is, light rays in $\mathbb{M}^{4}$ too.

Let us consider the Cauchy surface in $S_{1}^{3}$ given by $C_S=\overline{C} \cap S_{1}^{3}$, that is, the 2-surface satisfying 
$$
\left\{ 
\begin{array}{l}
t=0 \\ 
x^2 +y^2 +z^2 =1
\end{array}
\right.   
$$
so we can parametrize $C_S$ by 
\begin{equation}\label{Cparam}
\left\{ 
\begin{array}{l}
x=\cos u \sin w \\ 
y=\sin u \sin w \\ 
z=\cos w
\end{array}
\right.   
\end{equation}

Obviously, the null geodesic $\gamma\in \overline{\mathcal{N}}$ will entirely lie in $S^{3}_{1}$ if it satisfies equation (\ref{Secuacion}), so for every $s$ we have
\[
-s^2 + \left(x+s\cos \theta \sin \phi \right)^{2} + \left(y +s\sin \theta \sin \phi\right)^{2} + \left(z+s\cos \phi\right)^{2} = 1  \, ,
\]
which can be simplified into
\[
2s\left(\left(x\cos \theta +y\sin \theta\right)\sin \phi+z\cos \phi\right)=0  \, ,
\]
therefore 
\begin{equation}  \label{cond_v1}
\left(x\cos \theta +y\sin \theta\right)\sin \phi+z\cos \phi=0 \, ,
\end{equation}
and hence, we solve 
$$
\cot \phi=-\frac{x \cos \theta+y \sin \theta}{z} \, .
$$%
By the relation (\ref{Cparam}) we can write
$$
\cot \phi=-\cos \left(\theta-u\right) \tan w  
$$%
so $\phi$ only depends on the variables $u,w,\theta$. 
We will abbreviate it as
$$
\cot \phi=f\left(u,w,\theta\right)  
$$

Let us restrict the contact form $\alpha$ to $\mathcal{N}_S$ using:
\begin{equation} \label{restriccion_1}
\left\{ 
\begin{array}{l}
x = \cos u \sin w \\ 
y = \sin u \sin w \\ 
z = \cos w \\
\theta = \theta \\
\phi = \mathrm{arccot} f\left(u,w,\theta\right)%
\end{array}%
\right.   
\end{equation}%
Substituting the differentials
$$
\left\{ 
\begin{array}{l}
dx = -\sin u \sin w \, du + \cos u \cos w \, dw \\ 
dy = \cos u \sin w \, du + \sin u \cos w \, dw \\ 
dz = -\sin w \ dw\\
\end{array}%
\right.   
$$%
into $\overline{\alpha}$, we get:
\begin{equation}
\alpha_S = \left.\overline{\alpha}\right|_{\mathcal{N}_S}=   \frac{-\cos w \sin w \sin\left(\theta-u\right)}{\sqrt{\cos^{2}\left(\theta-u\right)\sin^{2} w + \cos^{2} w}} du - \frac{\cos\left(\theta-u\right)}{\sqrt{\cos^{2}\left(\theta-u\right)\sin^{2} w + \cos^{2} w}}  dw     
\end{equation}
where we have used the relations, obtained from (\ref{cond_v1}), given by
\begin{equation} \label{trigon-phi}
 \sin \phi = \frac{-\cos w}{\sqrt{\cos^{2}\left(\theta-u\right)\sin^{2} w + \cos^{2} w}} \, , \quad    \cos \phi = \frac{\sin w \cos\left(\theta-u\right)}{\sqrt{\cos^{2}\left(\theta-u\right)\sin^{2} w + \cos^{2} w}}  \, . \end{equation}

Then we can choose the following contact form in $\mathcal{N}_S$ 
$$ 
\alpha_S = \cos w \sin w \sin\left(\theta-u\right) du + \cos\left(\theta-u\right) dw \, ,
$$
and  the 2-plane that annihilates $\alpha_S$ is 
$$ 
\left(\mathcal{H}_S\right)_{\gamma} = \mathrm{span} \left\{  \textstyle{     -\cos\left(\theta-u\right)\left(\frac{\partial}{\partial u}\right)_{\gamma} + \cos w \sin w \sin\left(\theta-u\right) \left(\frac{\partial}{\partial w}\right)_{\gamma} , \left(\frac{\partial}{\partial \theta}\right)_{\gamma}     }      \right\}
$$

In order to find the future $l$-boundary of $3$--dimensional de Sitter space-time, in virtue of Section \ref{sec:embedded}, we will just restrict the results obtained in Section~\ref{sec:Mink-4} for $\mathbb{M}^{4}$ to the embedded $S^{3}_{1}$. 
So, using the expression (\ref{restriccion_1}) for the values $\left(u_0,w_0,\theta_0\right)$ we get: 
\[
\left(x_0, y_0, z_0, \theta_0, \phi_0\right) = \left(\cos u_0 \sin w_0, \, \sin u_0 \sin w_0, \, \cos w_0, \, \theta_0, \, \mathrm{arccot} f\left(u_0,w_0,\theta_0\right)\right)
\]
and substituting it, together with (\ref{trigon-phi}), into the equation (\ref{plane-orbit-boundary}), we obtain the equation of the orbit $\left(X_{S}^{+}\right)_{\gamma} = \overline{X}^{+}_{\gamma} \cap \mathcal{N}_S$ of $\oplus_{S}$ through $\gamma$ as a curve in the Cauchy surface $C_S$ given by
\begin{equation}\label{dS-orbit-boundary}
\cos \left(\theta_0-u\right)\tan w = \cos \left(\theta_0-u_0\right)\tan w_0
\end{equation}
or equivalently
\begin{equation}\label{dS-orbit-boundary-2}
f\left(u,w,\theta_0\right)=f\left(u_0,w_0,\theta_0\right).
\end{equation}

If we consider the inclusion in coordinates 
\begin{equation}\label{inc-Ns-N}
\begin{tabular}{l}
$i:\mathcal{N}_S \simeq \mathbb{S}^{2}\times\mathbb{S}^{1} \rightarrow \overline{\mathcal{N}}\simeq \mathbb{R}^{3}\times\mathbb{S}^{2}$ \\
$\left(u,w, \theta \right) \mapsto \left(\cos u \sin w, \, \sin u \sin w, \, \cos w, \, \theta , \, \mathrm{arccot} f\left(u,w,\theta\right)\right)$ 
\end{tabular}
\end{equation}
then its composition with the map (\ref{map-boundary}) is
\begin{equation}\label{dS-map-boundary}
\begin{tabular}{rcl}
$\mathcal{N}_S \simeq \mathbb{S}^{2}\times\mathbb{S}^{1}$ & $\rightarrow$ & $\partial^{+}\Sigma_S \subset \mathbb{R}^{1}\times\mathbb{S}^{2} $\\
$\left(u,w, \theta \right)$ & $\mapsto$ & $\left(0, \theta , \mathrm{arccot} f\left(u,w,\theta\right)\right)$
\end{tabular}
\end{equation}
For a fixed $\theta=\theta_0$, because (\ref{dS-orbit-boundary-2}), every level set $U_k = \left\{ \left(u,w\right)\in C_S: f\left(u,w,\theta_0\right)=k\right\}$ corresponds to an orbit of $\oplus_{S}$.
Since the image of 
\[
F\left(u,w\right)= f\left(u,w,\theta_0\right)= -\cos \left(\theta_0-u\right)\tan w 
\]
is $\left(-\infty,\infty\right)$ then the image of 
\[
G\left(u,w\right)= \mathrm{arccot} f\left(u,w,\theta_0\right) 
\]
is $\left(0,\pi\right)$, therefore the image of the map (\ref{dS-map-boundary}) is $\partial^{+}\Sigma_S = \lbrace 0 \rbrace \times\mathbb{S}^{2}\simeq \mathbb{S}^{2}$. 

By \cite[Prop. 4.28]{On83}, it can be easily observed that $I^{-}\left(p\right) \cap S^3_1 = I^{-}\left(p,S^3_1\right)$ and hence, for any light ray $\gamma\in \mathcal{N}_S $
\[
I^{-}\left(\gamma\right) \cap S^3_1 = I^{-}\left(\gamma,S^3_1\right)  \, .
\]
Thus, the restriction of TIPs of $\mathbb{M}^4$ to de Sitter space-time  are TIPs of $S^3_1$, and therefore the future $l$-boundary of de Sitter space-time coincides again with the part of the future $c$-boundary accessible by null geodesics.


\subsection{A family of $3$--dimensional space-times}\label{sec:alphaM}

In this section we will study the family of space-times given by $M_{\alpha}=\left\{\left(t,x,y\right)\in \mathbb{R}^3:t>0 \right\}$ with metric tensor $\mathbf{g}_{\alpha}=-t^{2\alpha}dt\otimes dt + dx\otimes dx + dy\otimes dy$. 

It is trivial to see that the transformations given by 
\begin{equation}\label{conformal-diffeo}
\begin{tabular}{ccccc}
\underline{For $ \alpha < -1$}: & & \underline{For $ \alpha = -1$}: & & \underline{For $ \alpha > -1$}: \vspace{3mm} \\
$\left\{
\begin{array}{l}
\overline{t} = \frac{t^{\alpha + 1}}{\alpha + 1} \\
\overline{x} = x \\
\overline{y} = y
\end{array}
\right.$
& &  
$\left\{
\begin{array}{l}
\overline{t} = \log t \\
\overline{x} = x \\
\overline{y} = y
\end{array}
\right.$
& &
$\left\{
\begin{array}{l}
\overline{t} = \frac{t^{\alpha + 1}}{\alpha + 1} - 1 \\
\overline{x} = x \\
\overline{y} = y
\end{array}
\right.$
\end{tabular}
\end{equation}
are conformal diffeomorphisms such that
$$
\begin{tabular}{ccccc}
\underline{For $ \alpha < -1$}: & & \underline{For $ \alpha = -1$}: & & \underline{For $ \alpha > -1$}: \vspace{3mm} \\
$M_{\alpha} \simeq \mathbb{M}^3$
& &  
$M_{-1} \simeq \mathbb{M}^3$
& &
$M_{\alpha} \simeq M_{*}$
\end{tabular}
$$
where the last space-time  $M_{*}$ denotes the $3$--dimensional Minkowski block studied in Section \ref{sec:Mink-3}. 
So, the space of light rays, its contact structure and the $l$-boundary of these space-times are already calculated in section~\ref{sec:Mink-3}. 

We will now examine the $l$-boundary for $\alpha>-1$.

Observe that the null vectors in $T_p M_{\alpha}$ are proportional to $v=\left(1,t^{\alpha}\cos \theta , t^{\alpha}\sin \theta\right)$ for $\theta\in \left[0,2\pi\right]$ at $p=\left(t,x,y\right)$, and the only non--zero Christoffel symbol is $\Gamma_{00}^{0}=\alpha  t^{-1}$. Hence, since the equations of geodesics are 
$$
\left\{
\begin{array}{l}
t'' + \frac{\alpha}{t}\left(t'\right)^2 = 0 \\
x'' = 0 \\
y'' = 0
\end{array}
\right.
$$
then the null geodesic $\gamma$ such that $\gamma\left(0\right)=\left(t_0, x_0, y_0  \right)$ and $\gamma'\left(0\right)=\left(1,t_{0}^{\alpha}\cos \theta_0 , t_{0}^{\alpha}\sin \theta_0  \right)$ for a given $\theta_0\in \left[0,2\pi\right]$ for $\alpha > -1$ can be written as
\[
\gamma\left(s\right)= \left( \left(\left(\alpha +1\right)t_{0}^{\alpha} s + t_{0}^{\alpha+1}\right)^{1/\left(\alpha + 1\right)} \, , \, x_0 + s t_{0}^{\alpha}\cos \theta_0 \, , \, y_0 + s t_{0}^{\alpha}\sin \theta_0\right)
\]
defined for $s\in \left(-\frac{t_0}{\alpha + 1}, \infty\right)$.

Observe that, when $-1<\alpha<0$, lightcones open wider as $t$ approaches to $0$, becoming a plane at the limit $t=0$. 
On the other hand, when $\alpha>0$, they close up when $t$ gets close to $0$, degenerating into a line when $t=0$. 
The case $\alpha=0$ corresponds to a Minkowski block isometric to $M_{*}$. 

Let us consider $C\equiv\left\{t=1\right\}$ as the global Cauchy surface we will use as origin of any given null geodesic 
\[
\gamma\left(s\right)= \left( \left(\left(\alpha +1\right) s + 1\right)^{1/\left(\alpha + 1\right)} \, , \, x_0 + s \cos \theta_0 \, , \, y_0 + s \sin \theta_0\right) = \left(t_s,x_s,y_s\right)
\]
Then the curve 
\[
\mu_{\theta}\left(\tau\right) = \left( \left(\left(\alpha +1\right)t_{s}^{\alpha} \tau + t_{s}^{\alpha+1}\right)^{1/\left(\alpha + 1\right)} \, , \, x_s + \tau t_{s}^{\alpha}\cos \theta \, , \, y_s + \tau t_{s}^{\alpha}\sin \theta\right)
\]
describes a null geodesic starting at $\gamma\left(s\right)$. 
So, for $\tau=\frac{-s}{t_s^{\alpha}}$, we have 
\[
\mu_{\theta}\left(-s/t_s^{\alpha}\right) = \left( 0 , \, x_0 + s \left(\cos \theta_0 - \cos \theta\right)  , y_0 + s \left(\sin \theta_0 - \sin \theta\right) \right) \in C .
\] 
Therefore, the coordinates of the sky of $\gamma\left(s\right)$ can be written by
\[
\psi\left(S\left(\gamma\left(s\right)\right)\right) \equiv \left\{
\begin{array}{l}
x\left(\theta\right) =  x_0 +s\left( \cos \theta_0  - \cos \theta   \right) \\
y\left(\theta\right) =  y_0 +s\left( \sin \theta_0  - \sin \theta   \right) \\
\theta\left(\theta\right) =  \theta 
\end{array}
\right. 
\]
Deriving with respect to $\theta$ at $\theta=\theta_0$, we obtain a generator of the tangent space of the sky $S\left(\gamma\left(s\right)\right)$ at $\gamma$, so
$$
T_{\gamma}S\left(\gamma\left(s\right)\right) = \mathrm{span}\left\{   \textstyle{    s\left( \sin\theta_0  \left( \frac{\partial}{\partial x} \right)_{\gamma} - \cos\theta_0  \left( \frac{\partial}{\partial y} \right)_{\gamma} \right) + \left( \frac{\partial}{\partial \theta} \right)_{\gamma}   }   \right\}
$$
and then
\[
\left(\ominus_{\alpha}\right)_{\gamma}=\lim_{s\mapsto \frac{-1}{\alpha + 1}}T_{\gamma}S\left(\gamma\left(s\right)\right) = \mathrm{span}\left\{    \textstyle{     -\sin\theta_0 \left( \frac{\partial}{\partial x} \right)_{\gamma} + \cos\theta_0 \left( \frac{\partial}{\partial y} \right)_{\gamma} + \left(\alpha + 1\right) \left( \frac{\partial}{\partial \theta} \right)_{\gamma}    }     \right\} \, .
\]

The solution $c\left(\tau\right)=\left( x\left(\tau\right), y\left(\tau\right), \theta\left(\tau\right) \right)$ of the initial value problem 
\[
\left\{
\begin{array}{l}
x'\left(\tau\right) =  -\sin \theta  \\
y'\left(\tau\right) =  \cos \theta  \\
\theta '\left(\tau\right) =  \alpha + 1 \\
c\left(0\right) = \left(x_0,y_0,\theta_0 \right)
\end{array}
\right. 
\]
describes the orbit $X^{-}_{\gamma}\subset \mathcal{N}_{\alpha}$ of $\ominus_{\alpha}$ passing by $\gamma$.
Then 
\[
c\left(\tau\right)=\left(x_0 + \frac{\cos \left(\left(\alpha + 1\right)\tau+ \theta_0\right)-\cos \theta_0}{\alpha+ 1}   \, , \, y_0 +   \frac{\sin \left(\left(\alpha + 1\right)\tau+ \theta_0\right)-\sin \theta_0}{\alpha+ 1}    \, , \, \left(\alpha + 1\right)\tau+ \theta_0\right) \, .
\] 

\begin{figure}[h]
  \centering
    \includegraphics[scale=0.4]{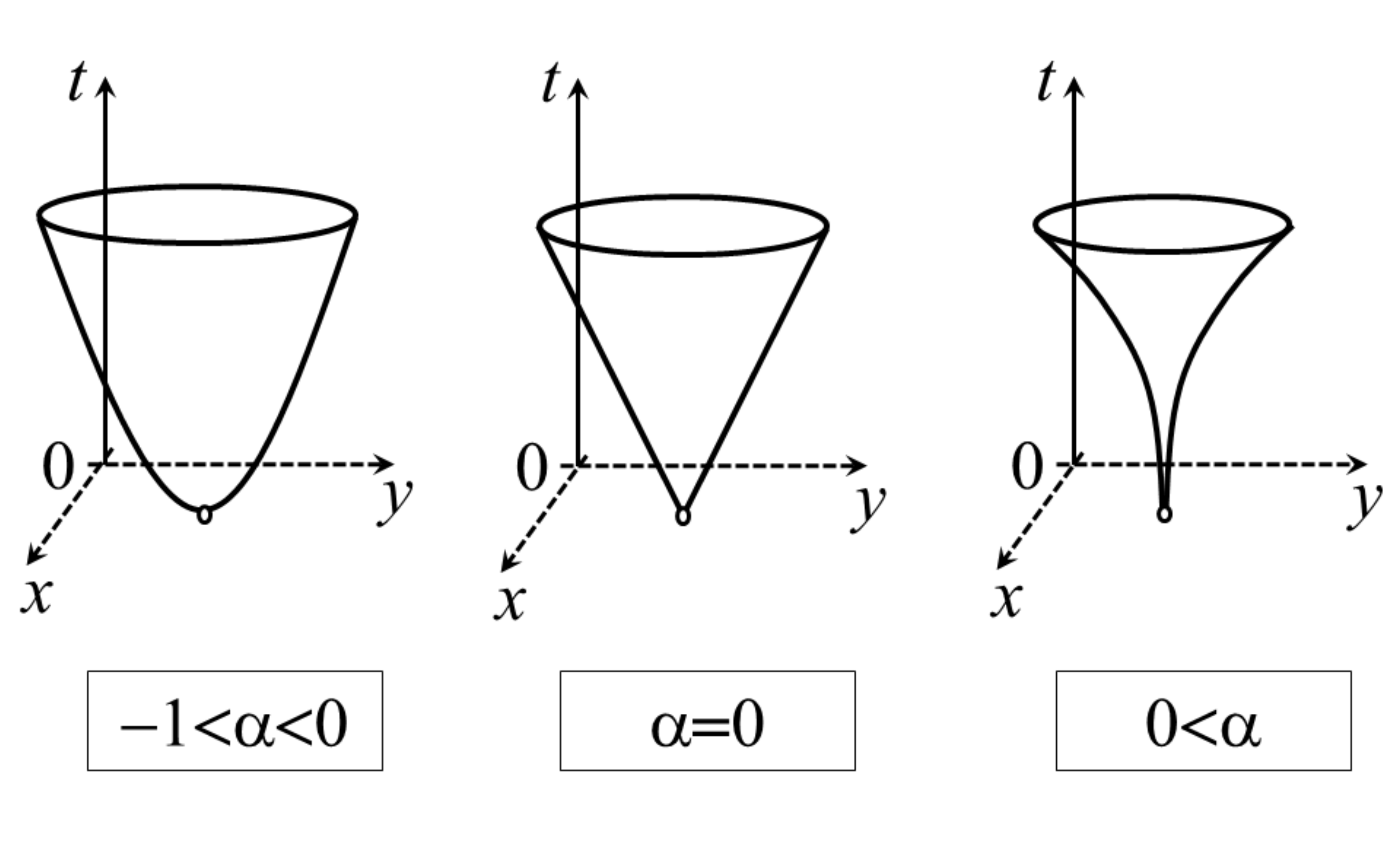}
  \caption{The $\alpha$-family of space-times.}
  \label{diapositiva6}
\end{figure}

It is easy to realize that the points in $M_{\alpha}$ in the orbit $X^{-}_{\gamma}$ verify 
\begin{equation}\label{conos-orbitas}
t^{2\alpha +2} = \left(\alpha + 1\right)^2 \left[    \textstyle{     \left( x - \left( x_0 - \frac{\cos \theta_0}{\alpha + 1} \right) \right)^2 + \left( y - \left( y_0 - \frac{\sin \theta_0}{\alpha + 1} \right) \right)^2     }     \right]
\end{equation}

A schematic picture of $X^{-}_{\gamma}$ can be seen in Figure \ref{diapositiva6}.

Observe that each orbit $X^{-}_{\gamma}$ is determined by the vertex of the surface (\ref{conos-orbitas}), therefore the past $l$-boundary can be identified with $\mathbb{R}^2$ such that any $\left(u,v\right)\in \mathbb{R}^2$ corresponds to the orbit of $\ominus_{\alpha}$ whose light rays emerges from the point $\left(t,x,y\right)=\left(0,u,v\right)$.

The differentiable structure of $\overline{M_{\alpha}}=M_{\alpha} \cup \partial^{-}\Sigma_{\alpha}$ cannot be the standard one induced from $\overline{M_{*}}=M_{*} \cup \partial^{-}\Sigma_{*}=\left\{\left(t,x,y\right)\in \mathbb{R}^3:t\geq -1 \right\}$ by the corresponding conformal mapping (\ref{conformal-diffeo}), because it would be needed that 
\[
\overline{M_{\alpha}}  \to  \overline{M_{*}} \, , \qquad
\left(t, x, y \right)  \mapsto  \left(\frac{t^{\alpha + 1 }}{\alpha + 1} -1 , x, y \right)
\] 
were differentiable, but it is not the case with the standard differentiable structure when $-1 < \alpha <0$.


\section{Conclusions and discussion}\label{sec:discussion}

The notion of a new causal boundary proposed by R. Low \cite{Lo06} and called $l$-boundary in this paper, which is based on the idea of determining all light rays which focus at the same point at infinity and treating this set as the `sky' of the common future endpoint of all of them, has been made precise and discussed carefully in the particular instance of three-dimensional space-times.        

It has been shown that under mild conditions, i.e., that the space $M$ doesn't have tangent skies, the regularity of the asymptotic distributions $\oplus$ and $\ominus$, and the smooth extension of the natural distribution $\widetilde{\mathcal{D}}$ on $\widetilde{\mathcal{N}}$ to its boundary, that such boundary $\partial \Sigma$ is well defined and makes the completed space $\overline{M}$ into a smooth manifold with boundary.  Let us point out here that the former condition can be removed as it will be shown elsewhere.   Space-times such that the $l$-boundary $\partial \Sigma$ exists and the completed space-time  $\overline{M} = M \bigcup \partial \Sigma$ is a smooth manifold with boundary could be called $l$-extendible. 

The $l$-boundary of a three-dimensional space-time  has been compared with the GKP $c$-boundary and it has been found that, even if in general the $l$-boundary is smaller, in the case that the conformal structure can be extended to the $l$-boundary the $l$-boundary and $c$-boundary are equivalent in the set where light rays are transversal.

Hence, a natural question emerges from the previous considerations: suppose that $M$ is a three-dimensional $l$-extendible space-time, can the conformal structure $\mathcal{C}$ on $M$ be smoothly extended to $\overline{M}$?

The answer to this question could seem to be negative.  Consider, for instance, the example $M_\alpha$, $\alpha = -1/2$, discussed in Sect. \ref{sec:alphaM} with representative metric $\mathbf{g} = -\frac{1}{t} dt \otimes dt + dx \otimes dx + dy\otimes dy$.  The space-time  $M_{-1/2}$ is conformally isometric to the block Minkowski space $M_*$ discussed in the second part of Section \ref{sec:Mink-3}, and we conclude that is $l$-extensible.  However it doesn't seem to be conformally extensible to the $l$-completed space $\overline{M}_{-1/2}$.  This apparent contradiction can be solved by noticing that the induced smooth structure on the $l$-completed space is not the one induced by the ambient smooth structure on $\mathbb{M}^3$.
It can be seen, the details will be discussed elsewhere, that there is a canonical projective conformal parameter on light rays such that the induced smooth structure on the boundary can be suitably described and the existence, or not, of a conformal extension to the $l$-boundary remains unanswered.



%
%

%




\section*{Acknowledgements}
The authors would like to thank the referee's comments and suggestions as well as the financial support provided by Ministry of Economy and Competitivity of Spain under the grant MTM2014-54692-P and Community of Madrid research project QUITEMAD+, S2013/ICE-2801.


\end{document}